\documentclass[letterpaper,12pt,reqno]{article}

\usepackage[letterpaper, left=2.5cm,right=2.5cm,top=2.5cm,bottom=2.6cm]{geometry}

\usepackage{setspace}

\usepackage[T1]{fontenc}
\usepackage{newtxtext}
\usepackage{microtype}

\usepackage[hyphens]{url}
\usepackage[hyphenbreaks]{breakurl}
\usepackage{hyperref}
\hypersetup{colorlinks=true, linkcolor=blue, urlcolor=blue, citecolor=black}

\usepackage{titlesec}
\usepackage[utf8]{inputenc}
\usepackage{graphicx}
\usepackage{amsmath}
\usepackage{MnSymbol}%
\usepackage{wasysym}%
\graphicspath{{Images/}}
\usepackage{amsfonts}
\usepackage{commath}
\usepackage{algpseudocode}
\usepackage{enumitem}
\setlist[enumerate]{label=$(\roman*)$,ref=(\roman*)}

\usepackage{xcolor}
\usepackage[english]{babel}
\usepackage{amsthm}
\usepackage{tocloft}
\usepackage{bbm}
\usepackage{mathtools}
\usepackage{tikz}
\usepackage{comment}
\usetikzlibrary{decorations.pathreplacing}
\usetikzlibrary{positioning,quotes,calc}
\usepackage{tikz-3dplot}

\usepackage{subcaption}
\usetikzlibrary{calc,intersections}

\definecolor{niceblue}{RGB}{0,104,178}
\definecolor{darkgreen}{RGB}{3,131,127}
\definecolor{orange}{RGB}{213,94,0}
\definecolor{darkorange}{RGB}{191,84,0}
\definecolor{paleniceblue}{RGB}{118,171,208}
\definecolor{camniceblue}{RGB}{145,182,177}
\definecolor{grey}{RGB}{60,62,61}
\definecolor{purple}{RGB}{153,51,255}
\definecolor{red}{RGB}{255,51,51}
\definecolor{darkdarkgreen}{RGB}{10,105,102}
\definecolor{darkred}{RGB}{148,47,47}

\definecolor{niceblue}{RGB}{0,104,178}
\definecolor{darkgrey}{RGB}{3,131,127}
\definecolor{orange}{RGB}{213,94,0}
\definecolor{darkorange}{RGB}{191,84,0}
\definecolor{paleniceblue}{RGB}{118,171,208}
\definecolor{camniceblue}{RGB}{145,182,177}
\definecolor{grey}{RGB}{60,62,61}
\definecolor{purple}{RGB}{153,51,255}
\definecolor{red}{RGB}{255,51,51}
\definecolor{darkdarkgrey}{RGB}{10,105,102}
\definecolor{darkred}{RGB}{148,47,47}

\definecolor{teal}{RGB}{0,150,136}        
\definecolor{peach}{RGB}{255,160,122}      
\definecolor{lavender}{RGB}{178,102,255}  
\definecolor{midnightniceblue}{RGB}{25,25,112} 
\definecolor{olive}{RGB}{128,128,0}       
\definecolor{mustard}{RGB}{204,153,0}      
\definecolor{rose}{RGB}{255,102,102}       
\definecolor{plum}{RGB}{102,0,102}         

\newtheorem*{theorem*}{Theorem}
\newtheorem*{corollary*}{Corollary}
\newtheorem*{proposition*}{Proposition}
\newtheorem*{lemma*}{Lemma}
\newtheorem*{fact*}{Fact}
\newtheorem*{definition*}{Definition}
\newtheorem*{conjecture*}{Conjecture}

\newtheorem{theorem}{Theorem}
\newtheorem{example}{Example}
\newtheorem{corollary}{Corollary}
\newtheorem{proposition}{Proposition}

\newtheorem{definition}{Definition}
\newtheorem{lemma}{Lemma}

\newtheorem{remark}{Remark}
\newtheorem{fact}{Fact}

\DeclarePairedDelimiterX{\inp}[2]{\langle}{\rangle}{#1, #2}

\DeclareMathOperator*{\argmax}{arg\,max}
\DeclareMathOperator*{\argmin}{arg\,min}

\newcommand{\indiffu}[1]{\mathrel{\underset{#1}{\Leftrightarrow}}}

\usepackage[font=small, width=0.8\textwidth]{caption}

\titleformat{\section}
        {\large\bfseries\center}
        {\thesection}
        {0.5em}
        {}
        []

\titleformat{\subsection}
        {\normalfont\bfseries}
        {\thesubsection}
        {0.5em}
        {}

\titleformat{\subsubsection}[runin]
        {\normalfont\bfseries}
        {\thesubsubsection}
        {0.5em}
        {}
        [.]

\AtBeginDocument{
  \setlength{\abovedisplayskip}{9pt}
  \setlength{\belowdisplayskip}{9pt}
  \setlength{\abovedisplayshortskip}{1pt}
  \setlength{\belowdisplayshortskip}{6pt}
  \setlength{\jot}{5pt}
}

\usepackage[authoryear]{natbib}
\usepackage{soul}
\title{Strategically Analogous Mechanisms\thanks{We thank Mohammad Akbarpour, Roberto Corrao, Laura Doval, Piotr Dworczak, Shengwu Li, Paul Milgrom, Lea Nagel, Michael Ostrovsky, Roberto Saitto, Ilya Segal, Andrzej Skrzypacz, and Takuo Sugaya for helpful comments.}}
\author{%
\begin{tabular}{cc}
Joseph Feffer & Filip Tokarski \\
Stanford GSB & Stanford GSB
\end{tabular}%
}
\begin{document}
\date{} 
\maketitle
\vspace*{-1cm} 

\begin{abstract}
This paper studies when strategic understanding acquired in one mechanism can be transferred to another. It introduces a framework in which agents have a limited understanding of how actions affect payoffs, and formalizes how they use \emph{analogies} to transfer this understanding between mechanisms. We ask when such reasoning allows agents to transfer their understanding of equilibrium between mechanisms. We first consider mechanisms that are \emph{strategically equivalent}---that is, share the same game form up to relabeling of actions---and show that agents' understanding of equilibrium transfers once they recognize how the actions in these mechanisms map to one another. We then define \emph{strategic analogy}, a weaker notion that allows actions and types to be remapped, and show that equilibrium understanding transfers once agents recognize how the actions and types in these mechanisms relate. Applications include single-item auctions, scoring auctions, and nonlinear pricing with capacity constraints.
\end{abstract}

\section{Introduction}

Strategizing in unfamiliar economic settings can be difficult even for sophisticated agents. However, when agents recognize that a setting is structurally related to one they already understand, they can overcome this difficulty by drawing on their existing strategic knowledge. In particular, agents facing a new mechanism may recognize that its strategic structure mirrors that of a known mechanism and use this correspondence to formulate a strategy in the new setting based on how they played in the old one. This possibility is most apparent when the two mechanisms are \emph{strategically equivalent}; that is, they share the same game form but present choices differently. A standard example is provided by the sealed-bid first-price auction and the Dutch auction, in which a price clock descends until a bidder stops it and wins the object at the current price \citep{vickrey1961counterspeculation,milgrom1982theory}. Despite their different presentations, both formats effectively ask the bidder to choose the price at which she is willing to buy. In the first-price auction, she states this price directly as her bid; in the Dutch auction, she waits for the clock to reach this price and then stops it. Recognizing this correspondence allows a bidder who knows how to optimally shade her bid in one format to apply that knowledge to the other.

The transfer of strategic reasoning, however, need not be limited to cases in which two mechanisms are just alternative representations of the same game. Indeed, some genuinely distinct mechanisms still require agents to reason through analogous tradeoffs and apply the same strategic principles. Consider, for example, an agent who understands how to play in a first-price auction with a specific reserve price. To bid optimally, she must be able to weigh the probability of winning against her conditional expected surplus, translate her beliefs about others' values into beliefs about the bids she will face, and understand how the reserve affects participation. Now, suppose that same agent were asked to play in a first-price auction with a different reserve. Changing the reserve affects her optimal strategy but not the underlying strategic considerations; an agent who recognizes this shared structure may therefore realize that she can choose how much to shade her bid by reasoning analogously to how she did in the original auction. In this paper, we develop a notion of \emph{strategic analogy} that generalizes this idea.

Understanding which mechanisms are strategically analogous is also important in practice, especially in settings where the same agent interacts with many versions of the same underlying problem. An example of such a setting comes from the Chilean public procurement system for pharmaceuticals studied by \cite{ALLENDE2024103086}. There, hospitals purchase supplies by holding procurement auctions on a public platform: each auction specifies a scoring rule that aggregates price and other product attributes into a single score, which is then used to determine the winner. Crucially, hospitals can express their preferences over price, delivery terms, and other product characteristics by customizing the weights the scoring rule puts on them. This also means, however, that manufacturers and wholesalers may bid in hundreds of auctions with potentially different strategic properties. In order to mitigate this problem, the platform fixes a baseline format and restricts hospitals to fine-tuning its parameters. However, parameter changes that adjust a mechanism's rules in a simple way may nevertheless alter its strategic properties. To minimize the need for bidders to repeatedly re-learn how to play, it is therefore important to understand which templates and parameter choices preserve the underlying strategic logic, and which ones generate genuinely different strategic problems.

This paper makes two main contributions. First, it develops a theory of reasoning transfer across strategic settings. To formalize this phenomenon, we model an agent's understanding of mechanisms through a set of payoff comparisons she can make, as well as her knowledge of which comparisons others are aware of. Our framework lets us define when a strategy profile is commonly known to be an equilibrium and captures how agents can transfer knowledge across mechanisms through analogies. We use the framework to study how similarities in mechanisms' rules help agents identify equilibria: Theorem~\ref{th:strategic-equivalence-transfer} shows that equilibrium understanding transfers across strategically equivalent mechanisms once agents recognize how actions correspond between them.

Our second contribution is the notion of \emph{strategic analogy}, which generalizes the idea of strategic equivalence. Intuitively, two mechanisms are strategically analogous if, after appropriately relabeling actions and remapping types, the payoffs from any action profile in one mechanism correspond with the payoffs from the analogous profile in the other. Allowing types to be remapped alongside actions provides a sense in which two mechanisms may share the same strategic structure even when they do not share the same game form. Theorem~\ref{th:strategic-analogy-transfer} then shows that, once the relevant correspondences between types and actions are explained, agents can transfer their understanding of equilibria across strategically analogous mechanisms.

We apply the framework in three settings. For single-item auctions, we show that \(k\)th-price auctions with reserves form a maximal class of strategically analogous mechanisms within a natural set of auctions. This aligns with common practice in auction houses and ad exchanges, which often fix the format and use the reserve price as the main tuning parameter. We also show that one strategically analogous class of symmetric mechanisms can implement the revenue-optimal auction for every symmetric regular prior, but not for every symmetric prior. Indeed, symmetric revenue-optimal auctions for irregular priors sometimes require choosing the winner at random.  The fact that bidders may be familiar with standard auctions but not with these randomized formats therefore provides a reason why auctioneers do not use them in practice, despite their theoretically higher expected revenue. The next application studies scoring auctions where a buyer awards a contract based on a score that aggregates price and quality. We show that varying the weight on price preserves strategic analogy under a linear scoring rule, but not under a rule in which the price score depends on how one's price compares to the lowest competing price. Finally, we study nonlinear pricing of services produced using scarce capacity. We show that charging for the underlying input can preserve strategic analogy across changes in productivity in cases where charging for realized output does not.

The rest of the paper is organized as follows. Section~\ref{sec:related-literature} discusses related literature. Section~\ref{sec:model} presents the model and introduces the concepts of payoff comparisons, knowledge, and knowledge transfer. Section~\ref{sec:strategic-equivalence} studies strategic equivalence and shows how explaining action correspondences lets agents transfer knowledge of equilibria across equivalent mechanisms. Section~\ref{sec:strategic-analogy} defines strategic analogy, establishes a knowledge transfer result, and applies the concept to single-item auctions, scoring auctions, and input- versus output-based pricing. Section~\ref{sec:discussion} then discusses potential extensions of our definitions. Finally, in Appendix~\ref{app:rules_based}, we discuss how strategic analogy can be communicated to agents through explaining similarities in mechanisms' rules.

\section{Related literature}\label{sec:related-literature}

Our paper builds on experimental work on reasoning transfer in strategic environments. \citet{RICK2010716} show that when subjects learn a general solution method such as iterated dominance, they are often able to apply it in new but structurally related environments. Similarly, \citet{cooper2008learning} find that experience in one signaling game improves equilibrium-consistent play in related games. In auction settings, \citet{harstad2000dominant} documents partial transfer of experience across related auction formats. \citet{breitmoser2022obviousness} find that performance in a sealed-bid second-price auction improves substantially when subjects are told that their bids determine dropout prices in a subsequent simulated clock auction. According to the authors, subjects find the dynamic setting easier to navigate, and showing them the clock implementation makes the analogy between the sealed-bid auction and the dynamic format more salient. This, in turn, helps participants transfer their strategic understanding to the sealed-bid format.

The effort to model agents' reasoning about mechanisms also connects to the epistemic game theory literature, which formalizes agents'  knowledge and uses it to justify solution concepts. The rationalizability literature characterizes
strategies consistent with common knowledge of rationality
\citep{bernheim1984rationalizable, pearce1984rationalizable,battigalli2003rationalizable}, while
\citet{aumann1995epistemic} identify epistemic conditions under which conjectures constitute a Nash equilibrium. We use similar tools to ask a distinct question: we assume agents can identify an equilibrium in one mechanism and ask under what conditions this knowledge lets them identify a corresponding equilibrium in another. Our model of agents' limited understanding of payoff comparisons also connects to work on imperfect knowledge of game structure: \citet{feinberg2021games} models unawareness through a collection of games representing players' perceptions of the strategic situation; \citet{copic2006awareness} define awareness equilibria for normal-form games in which players may be unaware of some available actions or possible types.

Strategic analogy also connects to the study of game isomorphisms and payoff transformations. \citet{mckinsey2016isomorphism} asks when games share the same strategic structure up to relabeling;  \citet{elmes1994strategic} extend related ideas to extensive-form games. \citet{moulin1978strategically} define strategic equivalence by preservation of preference rankings over mixed strategies, while \citet{morris2005generalized} and \citet{tewolde2021game} study payoff transformations that preserve best responses or Nash equilibria. We extend these ideas to incomplete-information settings, which lets us apply them beyond matrix games to mechanism design environments.

Another related strand of work models reasoning by analogy in decision theory and games. \citet{gilboa1995case} propose a model of choice under uncertainty in which agents look to similar past cases when choosing among actions. \citet{samuelson2001analogies} similarly treats analogies as a way of economizing on scarce reasoning resources: agents maintain a costly stock of models and apply the one that appears most suitable to the strategic interaction at hand. \citet{JEHIEL200581} develops the concept of analogy-based expectation equilibrium, in which agents bundle contingencies into analogy classes when forming expectations and best respond to those simplified representations. While these papers treat analogical reasoning as a form of bounded rationality, we instead study agents who respond optimally but can more easily identify the right strategy in mechanisms that are strategically close to ones they already understand.

Lastly, this paper connects to a literature studying mechanisms that are simple for participants to play. The practical importance of simplicity is emphasized by \citet{Shengwu_JEP}, who highlights that desirable theoretical properties of mechanisms can break down when agents fail to recognize incentives and respond to them correctly. Various approaches have been proposed to conceptualize such cognitive limitations. \citet{li2017obviously} asks when equilibrium play is transparent to agents with limited contingent reasoning. \citet{pycia2023theory} develop a theory of simplicity in extensive-form environments with limited planning horizons. \citet{borgers2019strategically} focus on belief sophistication, defining strategically simple mechanisms as those in which an agent can identify an optimal action using only first-order beliefs about others' preferences. While this literature studies how difficult it is to play a particular mechanism, we study how difficult it is to play within a \emph{class} of mechanisms, given that the agent already understands one of them. The perspective is complementary: in many applications, mechanisms should be simple not only in isolation, but also across related instances, so agents need not re-learn how to play when the environment changes. In this sense, we are aligned with \citet{brooks2025simplicity}, who argue that practically useful mechanisms should be portable across settings.

\section{Model}\label{sec:model}

\subsection{Environments, mechanisms, and equilibria}

We focus on the analogy of mechanisms designed for the same kind of economic problem. To make this precise, we formally define an \emph{environment} as a tuple $(\mathcal I,\mathcal Y,\mathcal T,\mathcal U),$ where \(\mathcal I\) is a finite nonempty set of agents, \(\mathcal Y\) is an outcome space, $\mathcal T=\prod_{i\in\mathcal I}\mathcal T_i$ is the type space, and $\mathcal U=(u_i)_{i\in\mathcal I}$ is the profile of Bernoulli utility functions giving the payoff each type of every agent gets from every outcome: $u_i:\mathcal T_i\times\mathcal Y\to\mathbb R.$\footnote{Note that types \(t_i\in\mathcal T_i\) are \emph{payoff types}: they specify preferences over outcomes, rather than information types or hierarchies of beliefs in the sense of \citet{harsanyi1967games}
and \citet{mertens1985formulation}. Thus, the assumption that an agent's utility depends only on her own type is substantive.} Agents maximize expected utility.

In each such environment, many mechanisms may be used to select an outcome based on agents' actions. We define a \emph{mechanism} for an environment $(\mathcal I,\mathcal Y,\mathcal T,\mathcal U)$ as a pair \(X=(\mathcal A,\Phi)\), where $\mathcal A=\prod_{i\in\mathcal I}\mathcal A_i$ specifies agents' action spaces and $\Phi:\mathcal A\to\Delta(\mathcal Y)$ is the outcome rule, mapping each action profile to a lottery over outcomes. When \(\Phi(a)\) is deterministic at some outcome \(y\in\mathcal Y\), we abuse notation and write $\Phi(a)=y.$ For \(i\in\mathcal I\), we also write
\[
\mathcal A_{-i}:=\prod_{j\in \mathcal I\setminus\{i\}}\mathcal A_j,
\qquad
\mathcal T_{-i}:=
\prod_{j\in \mathcal I\setminus\{i\}}\mathcal T_j.
\]
We illustrate the distinction between environments and mechanisms with an example.

\begin{example}\label{example:1}
A single good can be allocated to one of two agents, or remain unassigned. Each agent \(i\) has value \(t_i\in\mathbb R\) for the good. This environment can be written as $\left(\mathcal I,\mathcal Y,\mathbb R^{\mathcal I},(u_i)_{i\in\mathcal I}\right),$ where $\mathcal I=\{1,2\},$ $\mathcal Y:=\mathcal I\cup\{\varnothing\},$ and $u_i(t_i,y)= t_i\,\mathbf{1}\{y=i\}.$

Consider two mechanisms for this environment. In the first mechanism, each agent reports whether she wants the good. If exactly one agent says yes, she gets it; otherwise the good is discarded:
\[
X_1=\left(\mathcal A_1^1\times \mathcal A_2^1,\Phi^1\right),
\ \
\text{where}
\ \
\mathcal A_i^1=\{\text{yes},\text{no}\},
\ \ \
\Phi^1(a_1,a_2)=
\begin{cases}
1 & \text{if } (a_1,a_2)=(\text{yes},\text{no}),\\
2 & \text{if } (a_1,a_2)=(\text{no},\text{yes}),\\
\varnothing & \text{otherwise.}
\end{cases}
\]
In the second mechanism, agents have no choice and the good is assigned uniformly at random:
\[
X_2=
\left(\mathcal A_1^2\times \mathcal A_2^2,\Phi^2\right),
\ \ 
\text{where}
\ \ 
\mathcal A_i^2=\{\ast\},
\ \ \
\Phi^2(\ast,\ast)=\tfrac12\delta_1+\tfrac12\delta_2.
\]
\end{example}

We now define the \emph{Bayes--Nash equilibrium} of a mechanism. For any mechanism \(X=(\mathcal A,\Phi)\), let
\[
U_j^X[t_j,A]
:=
\mathbb E_{m\sim A}\left[
\int_{\mathcal Y} u_j(t_j,y)\,d\Phi(m)(y)
\right]
\]
denote the expected utility agent \(j\) of type \(t_j\) gets when the mixed action profile \(A\in\Delta(\mathcal A)\) is played.\footnote{Throughout, outcome, type, and action spaces are standard Borel spaces, utility functions are measurable, and outcome rules and strategies are measurable probability kernels. All action, type, and outcome bijections are measurable with measurable inverses. We restrict attention to payoff situations and mixtures for which the relevant expected utilities are finite.} For a pure action profile \(a\in \mathcal A\), we abuse notation and write
\[
U_j^X[t_j,a]
:=
U_j^X[t_j,\delta_a].
\]
Now, let a \emph{strategy profile} $\sigma=(\sigma_i)_{i\in\mathcal I}$ be a collection of maps \(\sigma_i:\mathcal T_i\to \Delta(\mathcal A_i)\) specifying a lottery over available actions for each type of agent $i$. Given a strategy profile \(\sigma\) and a prior \(F\in \Delta(\mathcal T)\) over all agents' types, let \(A_{-i}(\sigma,F\mid t_i)\in \Delta(\mathcal A_{-i})\) denote the induced distribution over other agents' actions, given \(\sigma\) and \(F\), conditional on agent \(i\)'s type being \(t_i\).

\begin{definition}
A strategy profile \(\sigma\) is a \textbf{Bayes--Nash equilibrium} of mechanism \(X= (\mathcal A ,\Phi)\) under prior \(F\) if for all \(i\in\mathcal I\) and \(t_i\in \mathcal T_i\),
\[
\sigma_i(t_i)\in 
\argmax_{A_i\in \Delta(\mathcal A_i)}
U_i^X\bigl[t_i,\ A_i\otimes A_{-i}(\sigma,F\mid t_i)\bigr].
\]
\end{definition}

Here \(\otimes\) forms the product distribution: \(A_j\otimes A_{-j}\in\Delta(\mathcal A)\) is the mixed action profile induced by independent randomization according to \(A_j\in\Delta(\mathcal A_j)\) and \(A_{-j}\in\Delta(\mathcal A_{-j})\).

\begin{remark}
Note that a mechanism is defined independently of agents' priors over others' types. As a result, the revelation principle does not generally associate an indirect mechanism with a single incentive-compatible direct-revelation counterpart. To see why this is the case, note that the ``meaning'' of an action in the indirect mechanism is fixed across priors---the rules of a first-price auction treat a bid of \(b\) the same way, regardless of the distribution of opponents' types. By contrast, in the revelation-principle construction, the action that a reported type maps to depends on the prior and the equilibrium being implemented. Indeed, in a first-price auction, the equilibrium bid of type \(\theta\)---which the direct mechanism submits on behalf of an agent reporting \(\theta\)---generally varies with the distribution of opponents' types. Incentive-compatible direct-revelation counterparts constructed for different priors thus generally correspond to different mechanisms, even though their action spaces are identical.
\end{remark}

\subsection{Knowledge of payoff comparisons}

We now fix an environment \((\mathcal I,\mathcal Y,\mathcal T,\mathcal U)\) and consider agents' understanding of the payoff structures of mechanisms designed for it. We express this understanding using \emph{payoff comparisons}. For agent \(j\), call a triple \((X,t_j,A)\) a \emph{payoff situation} if \(X=(\mathcal A,\Phi)\) is a mechanism, \(t_j\in\mathcal T_j\), and \(A\in\Delta(\mathcal A)\). A \emph{payoff comparison} is an ordered pair of payoff situations with the same
mechanism, agent, and type, written as
\[
U_j^X[t_j,A]\ge U_j^X[t_j,A'].
\]
Intuitively, this payoff comparison says that agent \(j\) of type \(t_j\)
gets a weakly higher payoff from \(A\) than from \(A'\) in mechanism \(X\).
We use \(\mathfrak R\) to denote the set of all possible payoff comparisons,
and \(\mathfrak R^{\mathrm{true}}\subseteq\mathfrak R\) to denote the subset of
comparisons whose associated inequality is true.

\begin{remark}
We emphasize that payoff comparisons capture an agent's knowledge of how different action profiles rank within a given mechanism, not her knowledge of how she values outcomes directly. In particular, consider two action profiles in mechanism $X$ and two action profiles in mechanism $X'$ such that the corresponding pairs induce the same outcomes. Our framework allows the agent to know which of the two profiles is better in $X$ while being unaware of which is better in $X'$. This reflects the fact that recognizing the consequences of actions requires understanding how the mechanism maps actions to outcomes.
\end{remark}

We capture which payoff comparisons each agent is aware of through an \emph{awareness profile} $\omega\in\Omega:=\prod_{i\in\mathcal I}2^{\mathfrak R}$:
\[
\omega=(\mathcal{R}_i^\omega)_{i\in\mathcal I},
\qquad
\mathcal{R}_i^\omega\subseteq\mathfrak R.
\]
Here \(\mathcal{R}_i^\omega\) is the set of payoff comparisons agent \(i\) is aware of under \(\omega\). We model agent \(i\)'s uncertainty about agents' awareness through a nonempty \emph{entertainment set}
\[
\mathcal E_i\subseteq\Omega,
\]
consisting of the awareness profiles she considers possible. We say that agent \(i\) \emph{knows} an event \(E\subseteq\Omega\) if \(\mathcal E_i\subseteq E\). We illustrate this in the following example.

\begin{example}
Consider the environment from Example~\ref{example:1}. Suppose agent \(2\) knows that agent \(1\) is aware that, in mechanism \(X_1\), she is always weakly better off saying ``yes'' than saying ``no'' whenever her type is nonnegative. This can be formalized as:
\[
\mathcal E_2
\subseteq
\bigcap_{t_1\ge 0,\ a_2\in\mathcal A_2^1}
\left\{
\omega\in\Omega:
\bigl(
U_1^{X_1}[t_1,(\text{yes},a_2)]
\ge
U_1^{X_1}[t_1,(\text{no},a_2)]
\bigr)\in \mathcal{R}_1^\omega
\right\}.
\]
That is, agent \(2\) only entertains awareness profiles in which agent \(1\) is aware of each of these comparisons.
\end{example}

To define common knowledge, we first specify which entertainment sets each agent could have. For each agent \(i\), let \(\mathfrak E_i\) be a nonempty family of possible entertainment sets:
\[
\mathfrak E_i\subseteq 2^\Omega\setminus\{\varnothing\},
\qquad
\mathcal E_i\in\mathfrak E_i.
\]
The families \((\mathfrak E_i)_{i\in\mathcal I}\) are primitives of the knowledge framework and are commonly known; the actual entertainment sets \((\mathcal E_i)_{i\in\mathcal I}\) need not be. This gives the following definition:

\begin{definition}\label{def:common-knowledge}
An event \(E\subseteq\Omega\) is \textbf{common knowledge} if
\[
\mathcal E\subseteq E
\qquad
\text{for every }i\in\mathcal I
\text{ and }\mathcal E\in\mathfrak E_i.
\]
In particular, a comparison \(r\in\mathfrak R\) is \textbf{common knowledge} if
\[
\mathcal E\subseteq
\bigcap_{j\in\mathcal I}
\{\omega\in\Omega:r\in\mathcal R_j^\omega\}
\qquad
\text{for every }i\in\mathcal I
\text{ and }\mathcal E\in\mathfrak E_i.
\]
\end{definition}

\begin{remark}
We treat the commonly known families \((\mathfrak E_i)_{i\in\mathcal I}\) as the primitives of the epistemic structure. Accordingly, we define common knowledge of agents' awareness of payoff comparisons through (commonly known) restrictions on these entertainment sets. This lets us avoid explicitly modeling agents' beliefs over others' entertainment sets and the resulting hierarchy of knowledge.

It is also worth noting that our entertainment sets \((\mathcal E_i)_{i\in\mathcal I}\) play the same structural role as information partitions in the standard model of \citet{aumann1976agreeing}. His information partition specifies, for every state \(\omega\), the set of states agent \(i\) considers possible when \(\omega\) is true. Here we retain only a single possibility set \(\mathcal E_i\subseteq\Omega\), relevant for reasoning about the true awareness profile. Thus, each agent is assigned a single set \(\mathcal E_i\subseteq\Omega\), rather than a state-dependent information cell \(\Pi_i(\omega)\).

However, unlike the information partitions in \citet{aumann1976agreeing}, these entertainment sets need not themselves be commonly known. Instead, we move the commonly known primitive of the epistemic structure one step further, to the families \((\mathfrak E_i)_{i\in\mathcal I}\) of possible entertainment sets.\footnote{One could similarly extend Aumann's framework by assigning each agent a family of information structures she considers possible. A commonly known information structure, as in Aumann's model, would correspond to every agent's family containing only that same structure.} This lets us maintain uncertainty about some aspects of what others know while making other aspects common knowledge. Indeed, our common-knowledge assumptions concern only particular features of agents' entertainment sets. For example, as we explain in due course, a designer's public explanation of analogies can be represented by requiring every entertainment set in each family to contain only awareness profiles in which all agents reason according to those analogies. This captures the idea that everyone takes these analogies into account, everyone knows that everyone does so, and so on, regardless of what agents previously knew.
\end{remark}

Using this framework, we can define what it means for it to be common knowledge that a strategy profile \(\sigma\) is a Bayes--Nash equilibrium.

\begin{definition}\label{def:knowingeq}
Fix a mechanism \(X=(\mathcal A,\Phi)\), a prior \(F\in\Delta(\mathcal T)\), and an equilibrium \(\sigma\) of \(X\) under \(F\). We say that it is \textbf{common knowledge that $\sigma$ is an equilibrium of $X$ at $F$} if, for every \(i,h,j\in\mathcal I\), \(\mathcal E\in\mathfrak E_i\), \(\omega\in\mathcal E\), \(t_j\in\mathcal T_j\), and \(A_j\in\Delta(\mathcal A_j)\),
\[
\Bigl(
U_j^X\bigl[t_j,\sigma_j(t_j)\otimes A_{-j}(\sigma,F\mid t_j)\bigr]
\geq
U_j^X\bigl[t_j,A_j\otimes A_{-j}(\sigma,F\mid t_j)\bigr]
\Bigr)
\in\mathcal R_h^\omega.
\]
Here \(A_{-j}(\sigma,F\mid t_j)\) is the distribution over opponents' actions induced by \(\sigma_{-j}\) and the conditional distribution of \(t_{-j}\) given \(t_j\).
\end{definition}
Intuitively, this means that it is common knowledge that, at strategy profile $\sigma$, no one has a profitable unilateral deviation.

\subsection{Knowledge transfer}

We now formalize how agents who understand certain correspondences between
mechanisms can transfer payoff comparisons between them. We express these
correspondences using \emph{payoff analogies}. For agent \(j\), a payoff
analogy is an ordered pair of payoff situations, written as
\[
(X,t_j,A)\indiffu{j}(X',t_j',A'),
\]
where \(X=(\mathcal A,\Phi)\) and \(X'=(\mathcal A',\Phi')\) are mechanisms,
\(t_j,t_j'\in\mathcal T_j\), \(A\in\Delta(\mathcal A)\), and
\(A'\in\Delta(\mathcal A')\). We use \(\mathfrak L\) to denote the set of all
possible payoff analogies. Intuitively, such an analogy says that, for
agent \(j\), the payoff situation \((X,t_j,A)\) corresponds to the payoff
situation \((X',t_j',A')\). The exact meaning of a payoff analogy comes
from how agents can use it to extend the comparisons they already know. To
capture this, we define when a comparison set \(\mathcal R\) is
\emph{closed under an analogy set} \(\mathcal L\):

\begin{definition}\label{def:L-closed}
Fix \(\mathcal L\subseteq\mathfrak L\). Let
\(\overline{\mathcal L} \supseteq\mathcal L \) be the smallest subset of $\mathfrak L$ satisfying the following properties:
\begin{enumerate}
\item \textbf{Symmetry.}
  $\bigl((X,t_j,A)\indiffu{j}(X',t_j',A')\bigr)\in \overline{\mathcal L}$
  if and only if
  $\bigl((X',t_j',A')\indiffu{j}(X,t_j,A)\bigr)\in \overline{\mathcal L}$.
\item \textbf{Mixtures.}
  Let $(S,\mu)$ be a probability space and let
  $A^s\in\Delta(\mathcal A)$, ${A'}^s\in\Delta(\mathcal A')$ be
  measurable in~$s$. If
  $\bigl((X,t_j,A^s) \indiffu{j} (X',t_j',{A'}^s)\bigr)
  \in \overline{\mathcal L}$
  for every $s\in S$, then
\[
\bigl(
  (X,t_j,\textstyle\int_S A^s\,d\mu)
  \indiffu{j}
  (X',t_j',\textstyle\int_S {A'}^s\,d\mu)
\bigr)
\in \overline{\mathcal L}.
\]
\end{enumerate}
Then the comparison set \(\mathcal R\subseteq\mathfrak R\) is
\textbf{closed under \(\mathcal L\)} if, whenever
\[
\bigl\{ (X,t_j,A^1) \indiffu{j} (X',t_j',A^3),\;
           (X,t_j,A^2) \indiffu{j} (X',t_j',A^4) \bigr\} \subseteq \overline {\mathcal L},
\]
we have $ \ \ \bigl(U_j^X[t_j,A^1] \ge U_j^X[t_j,A^2]\bigr) \in \mathcal{R}
  \quad \implies \quad
  \bigl(U_j^{X'}[t_j',A^3] \ge U_j^{X'}[t_j',A^4]\bigr) \in \mathcal{R}.$
\end{definition}

Thus, a set of comparisons $\mathcal R$ is closed under $\mathcal L$ if it contains all the inferences the agent could make using those analogies. More specifically, an agent can reason with $\mathcal L$ in two ways. First, she can use analogies to transfer knowledge across mechanisms: once an agent knows that two payoff situations correspond, she can carry payoff comparisons from one to the other. Second, she can extend the set of analogies themselves by reversing them (the symmetry property) and by combining them across mixtures over action profiles (the mixtures property). The latter is important because it ensures that if an analogy holds pointwise across action profiles, then it also holds after randomizing over them; this lets agents transfer knowledge about mixed actions.

Importantly, different collections of mechanisms will admit different sets of analogies that agents can safely use, in the sense that these analogies will never generate false comparisons from true ones. We capture this with the following definition:

\begin{definition}\label{def:valid-analogy-set}
A set \(\mathcal L\subseteq\mathfrak L\) is \textbf{valid} if, for every
\(\mathcal R\subseteq\mathfrak R^{\mathrm{true}}\), the smallest set containing
\(\mathcal R\) and closed under \(\mathcal L\) is contained in
\(\mathfrak R^{\mathrm{true}}\).
\end{definition}

Intuitively, an analogy set $\mathcal L$ is valid if applying its associated reasoning principles---as specified in Definition~\ref{def:L-closed}---to a true comparison always produces true comparisons. As it turns out, valid analogy sets admit a clean characterization.

\begin{proposition}\label{th1:eqchar}
An analogy set \(\mathcal L\) is valid if and only if, for every agent
\(j\), every two mechanisms \(X=(\mathcal A,\Phi)\) and
\(X'=(\mathcal A',\Phi')\), and every two types
\(t_j,t_j'\in\mathcal T_j\), there exists a positive affine function
\begin{equation}\label{eq:affineinverse}
\ell^j_{(X,t_j)\to (X',t_j')}:\mathbb R\to\mathbb R 
\quad \text{such that}
\quad 
\ell^j_{(X',t_j')\to (X,t_j)}
=
\left(
\ell^j_{(X,t_j)\to (X',t_j')}
\right)^{-1},
\end{equation}
and such that for all \(A\in\Delta(\mathcal A)\) and
\(A'\in\Delta(\mathcal A')\),
\[
\bigl((X,t_j,A)\indiffu{j}(X',t_j',A')\bigr)\in\mathcal L
\quad\Longrightarrow\quad
U_j^{X'}[t_j',A']
=
\ell^j_{(X,t_j)\to (X',t_j')}
\bigl(U_j^X[t_j,A]\bigr).
\]
\end{proposition}

The affine structure arises from requiring analogies to remain valid when combined through mixtures. Since expected payoffs under mixtures are convex combinations of the payoffs being mixed, the transformation relating analogous payoff situations must preserve comparisons generated by all such convex combinations, and hence must be positive affine. This parallels the familiar uniqueness of expected-utility representations up to positive affine transformations \citep{von2007theory,herstein1953axiomatic}.

Finally, we formalize agents' beliefs about how others use analogies by imposing conditions on their entertainment sets.
\begin{definition}
We say agent \(i\) believes that \textbf{agent \(j\) reasons according to the analogy set} \(\mathcal L\subseteq\mathfrak L\) if, for every \(\omega\in\mathcal E_i\), the comparison set \(\mathcal R_j^\omega\) is closed under \(\mathcal L\). We say it is \textbf{common knowledge that agents reason according to} \(\mathcal L\) if
\[
\mathcal E\subseteq
\bigcap_{j\in\mathcal I}
\left\{
\omega\in\Omega:
\mathcal R_j^\omega \text{ is closed under }\mathcal L
\right\}
\qquad
\text{for every }i\in\mathcal I
\text{ and }\mathcal E\in\mathfrak E_i.
\]
\end{definition}
That is, reasoning according to \(\mathcal L\) is common knowledge if it is common knowledge that every agent \(i\) only considers possible awareness profiles in which every agent \(j\)'s comparison set contains all inferences that can be drawn from the analogies in \(\mathcal L\).

\section{Strategic equivalence}\label{sec:strategic-equivalence}

We now apply our framework to strategically equivalent mechanisms. Strategic equivalence is a standard concept in the mechanism design literature: two mechanisms are strategically equivalent if one can be transformed into the other by relabeling each agent's actions, so that corresponding action profiles yield the same payoff consequences for all agents and types (see e.g.\ \citet{vickrey1961counterspeculation,milgrom1982theory,krishna2009auction}).\footnote{Some notions of strategic equivalence allow one of the equivalent mechanisms to have payoff-equivalent duplicate actions, requiring bijections only after such duplicates are removed. We do not allow for this in order to simplify notation, but the definition could be straightforwardly amended to accommodate it.} In our framework, this can be defined as follows.

\begin{definition}
Mechanisms \(X=(\mathcal A,\Phi)\) and \(X'=(\mathcal A',\Phi')\) are \textbf{strategically equivalent} if there exists a profile of bijections
\[
\alpha=(\alpha_i)_{i\in\mathcal I},
\qquad
\alpha_i:\mathcal A_i\to\mathcal A_i'
\]
such that, for every \(i\in\mathcal I\), type \(t_i\in\mathcal T_i\), and action profile \(a\in\mathcal A\),
\[
U_i^{X}[t_i,a]
=
U_i^{X'}[t_i,\alpha(a)].
\]
A collection of mechanisms is a \textbf{class of strategically equivalent mechanisms} if every pair of mechanisms in the collection is strategically equivalent.
\end{definition}

That is, relabeling each agent's action individually through \(\alpha_i\) maps every action profile in \(X\) to one in \(X'\) that yields the same payoff for all agents and all types. We illustrate this through the example of first-price and Dutch auctions whose strategic equivalence has been noted by \cite{vickrey1961counterspeculation}.\footnote{The Dutch auction is dynamic; here we consider its reduced normal form where each bidder's strategy specifies the time at which to stop the clock.}

\begin{example}[\citet{vickrey1961counterspeculation,milgrom1982theory}]
\label{example:fpa-dutch}
A single good can be allocated to one of \(N\geq 2\) agents, or remain unassigned. Each agent \(i\) has value \(t_i\in\mathbb R\) for the good, and agents can be asked to make payments. Formally, this environment is $\big(\mathcal I,\mathcal Y,\mathbb R^{\mathcal I},(u_i)_{i\in\mathcal I}\big),$ where $\mathcal I=\{1,\dots,N\},$ $\mathcal Y:=(\mathcal I\cup\{\varnothing\})\times \mathbb{R}_+^{\mathcal I}, $ and $u_i\bigl(t_i,(j,p)\bigr)
=
t_i\,\mathbf{1}\{j=i\}-p_i.$ That is, an outcome specifies which agent receives the good and the payment made by each agent.

The following two mechanisms for this environment are strategically equivalent.
\begin{enumerate}
    \item 
\textbf{First-price auction with bounded bids.}
Each agent submits a bid in the interval \([0,\bar b]\). The good is allocated to one of the agents submitting the highest bid, according to a common fixed tie-breaking rule, and the winner pays her bid. Formally, $X^{\mathrm{FPA}}
=
\left(\mathcal A^{\mathrm{FPA}},\Phi^{\mathrm{FPA}}\right),$ where
\[
\mathcal A^{\mathrm{FPA}}=[0,\bar b]^{\mathcal I},
\qquad
\Phi^{\mathrm{FPA}}(a)
=
\bigl(w(a),\, a_{w(a)}\, e_{w(a)}\bigr),
\]
where \(w(a)\) denotes the agent selected as the winner from among
$\argmax_{j\in\mathcal I} a_j$ and \(e_i\in\mathbb R^{\mathcal I}\) is the unit vector with a \(1\) in coordinate \(i\) and \(0\) elsewhere.
\item
\textbf{Dutch auction.}
The price starts at \(\bar b\) and falls linearly to \(0\) at speed \(r>0\). Each agent chooses a stopping time at which to claim the good. The good is allocated to one of the agents choosing the earliest stopping time, using the same tie-breaking rule as in the first-price auction, and the winner pays the clock price at that time. Formally, $X^{\mathrm{Dutch}}
=
\left(\mathcal A^{\mathrm{Dutch}},\Phi^{\mathrm{Dutch}}\right),$ where
\[
\mathcal A^{\mathrm{Dutch}}=\left[0,{\bar b}/{r}\right]^{\mathcal I},
\qquad \Phi^{\mathrm{Dutch}}(a)
=
\bigl(\tilde w(a),\,(\bar b-r\,a_{\tilde w(a)})\,e_{\tilde w(a)}\bigr),
\]
where \(\tilde w(a)\) denotes the agent selected as the winner from among $\argmin_{j\in\mathcal I} a_j.$
\end{enumerate}
\end{example}

To understand this equivalence, note that by deciding when to stop the clock in the Dutch auction, a player is effectively choosing the price at which to claim the object. This is in turn equivalent to choosing a bid in a first-price auction. Indeed, these two decisions of player \(i\) can be connected with the following bijection:
\[
\alpha_i^{\text{FPA} \to \text{Dutch}}(a_i)=\frac{\bar b-a_i}{r}.
\]
This maps a bid \(a_i\) in the first-price auction to the moment when the Dutch clock reaches that same price. Therefore, for any bid profile in the first-price auction, the corresponding stopping-time profile in the Dutch auction yields the same winner and the same payment.

\begin{remark}
An analogous strategic equivalence holds between the sealed-bid second-price auction and the ascending auction, provided the latter does not reveal dropout information. If dropouts were observable, bidders could condition their play on others' dropout decisions, giving rise to strategies in the ascending auction that have no counterpart in the sealed-bid one. This issue does not arise for the first-price and Dutch auctions, since the Dutch auction ends the moment any bidder stops, so a bidder cannot condition her play on any non-trivial information about others' actions.
\end{remark}
Indeed, experimental evidence suggests that drawing such correspondences between mechanisms' rules can help agents strategize. For example, \citet{breitmoser2022obviousness} find that performance
in a sealed-bid second-price auction improves substantially when
subjects are told that their bids will determine the prices at which they drop out of a subsequent ascending-clock auction. The authors argue that the clock presentation prompts participants
to use the ``continue versus stop'' reasoning of an ascending-clock
auction, even though they ultimately submit a sealed bid. Thus,
drawing an analogy to a setting participants find easier to navigate
lets them apply that reasoning in the sealed-bid format, reducing
bidding errors.

We now discuss how communicating such correspondences to agents helps them transfer knowledge about equilibria from one mechanism to another. Suppose that two mechanisms are strategically equivalent and that an equilibrium of one is common knowledge, in the sense of Definition~\ref{def:knowingeq}. If it is also common knowledge that agents reason according to the relevant action correspondences, then the corresponding equilibrium of the other mechanism is common knowledge as well. This is captured by the following result.

\begin{theorem}\label{th:strategic-equivalence-transfer}
Let \(X=(\mathcal A,\Phi)\) and \(X'=(\mathcal A',\Phi')\) be strategically equivalent through \(\alpha\).
\begin{enumerate}
\item If \(\sigma\) is an equilibrium of \(X\) for prior \(F\), then the strategy profile defined by \(\sigma_i'(t_i):=(\alpha_i)_\#\sigma_i(t_i)\) is an equilibrium of \(X'\) for the same prior.
\item If, in addition, this equilibrium of \(X\) is commonly known and it is common knowledge that agents reason according to the analogy set
\[
\big\{(X,t_i,a)\indiffu{i}(X',t_i,\alpha(a)):
i\in\mathcal I,\ t_i\in\mathcal T_i,\ a\in\mathcal A\big\},
\]
then the equilibrium \(\sigma'\) of \(X'\) is commonly known.
\end{enumerate}
\end{theorem}

Here, \((\alpha_i)_\#\sigma_i(t_i)\) denotes the pushforward of \(\sigma_i(t_i)\) under \(\alpha_i\), that is, the distribution over actions obtained by relabeling each action through \(\alpha_i\). This relabeling preserves equilibrium: since the payoff consequences of actions remain unchanged, any profitable deviation in \(X'\) would correspond to a profitable deviation in \(X\). Hence, if \(\sigma\) is an equilibrium of \(X\), then \(\sigma'\) is an equilibrium of \(X'\).

Moreover, the theorem shows that if it is common knowledge that agents reason according to these action correspondences, then common knowledge of the original equilibrium carries over to the relabeled equilibrium. Note that the analogy set in the theorem is always valid: since strategic equivalence requires $U_i^{X'}[t_i,\alpha(a)] = U_i^X[t_i,a]$, the condition of Proposition~\ref{th1:eqchar} is satisfied with the identity as the affine map for all agents $j$ and all types $t_j\in \mathcal T_j$. The substantive assumption in part~$(ii)$ is therefore not that these analogies are correct, but that agents \emph{recognize} them---that is, that the correspondence between the two mechanisms' rules has been explained to them, and that it is common knowledge that they reason accordingly.

The proof is in the appendix. Knowledge transfers because the reasoning closure in Definition~\ref{def:L-closed} allows agents to use the action analogies to transfer the payoff comparisons that verify equilibrium. Although these analogies cover only pure action profiles, the mixture property extends them to lotteries over actions, so the result also applies to mixed strategies.

Several features of this result deserve comment. First, what must be communicated to agents is the equivalence of the mechanisms' \emph{rules}, and no information about equilibrium behavior needs to be conveyed. This distinction is practically relevant: while explaining to agents how they should play---which types should choose which actions and why---is a complex task, a correspondence between mechanisms' rules is a concrete and relatively simple object that could be communicated by a designer switching between mechanisms.\footnote{In Appendix~\ref{app:rules_based}, we formalize this idea and extend it to the notion of strategic analogy, defined in the next section.}

Second, the result is agnostic about how agents came to know the equilibrium of the original mechanism---the way they optimize their actions or form beliefs. It simply says that if agents have somehow acquired the relevant knowledge---formalized as the payoff comparisons in Definition~\ref{def:knowingeq}---they do not need to repeat that process for the new mechanism.

Third, one may ask whether agents can also coordinate on the transferred equilibrium, particularly when multiple equilibria exist. While our framework does not model coordination directly, the communication of the rule-level correspondence between mechanisms provides a natural focal point: if agents understand that a certain strategy profile is the one they play in the original mechanism, and they understand how actions map across mechanisms, it is natural for them to expect that the corresponding profile will be played in the new one.

\section{Strategic analogy}\label{sec:strategic-analogy}

In the previous section we showed how understanding strategic equivalence allows agents to transfer knowledge across mechanisms. We now introduce a more permissive notion of \emph{strategic analogy}, which generalizes strategic equivalence in two ways. First, while strategic equivalence only allows agents to relabel actions when moving from one mechanism to another, strategic analogy also allows them to reinterpret the meanings of \emph{types}. Second, rather than requiring that utilities in the two mechanisms be equal after these remappings, it requires only that they be related through a valid analogy set in the sense of Definition~\ref{def:valid-analogy-set}. This preserves a sense in which agents can transfer payoff comparisons when payoffs themselves lack an objective cardinal interpretation.

\begin{definition}\label{def:strategicanalogy}
Let \(X=(\mathcal A,\Phi)\) and \(X'=(\mathcal A',\Phi')\) be mechanisms for
the same environment. The mechanisms \(X\) and \(X'\) are
\textbf{strategically analogous} if there exist profiles of bijections
\[
\alpha=(\alpha_i)_{i\in\mathcal I}, 
\quad 
\alpha_i:\mathcal A_i\to\mathcal A_i',
\qquad \text{and} \qquad
\tau=(\tau_i)_{i\in\mathcal I},
\quad
\tau_i:\mathcal T_i\to\mathcal T_i,
\]
such that the analogy set
\[
\mathcal L_{\alpha,\tau}^{X,X'}
:=
\Bigl\{
(X,\tau_i(t_i),a)
\;\indiffu{i}\;
(X',t_i,\alpha(a))
:
i\in\mathcal I,\ t_i\in\mathcal T_i,\ a\in\mathcal A
\Bigr\}
\]

is valid. A collection of mechanisms is a \textbf{class of strategically
analogous mechanisms} if every pair of mechanisms in the collection is
strategically analogous.
\end{definition}

Note that the definition requires all relabelings to be separable: actions are relabeled agent by agent, as in strategic equivalence. Here the same is also true of types: each agent's type is remapped individually, and independently of the relabeling of actions. While we do not explicitly model the difficulty of learning different relabelings, separability ensures that they are relatively low-dimensional objects and may be easier for agents to understand.

We also present a characterization of strategic analogy that follows directly from Proposition~\ref{th1:eqchar}.

\begin{corollary}\label{prop:strategicanalogy-affine}
Let \(X=(\mathcal A,\Phi)\) and \(X'=(\mathcal A',\Phi')\) be
mechanisms for the same environment. Then \(X\) and \(X'\) are
strategically analogous if and only if there exist profiles of
bijections
\[
\alpha=(\alpha_i)_{i\in\mathcal I}, 
\quad 
\alpha_i:\mathcal A_i\to\mathcal A_i',
\qquad \text{and} \qquad
\tau=(\tau_i)_{i\in\mathcal I},
\quad
\tau_i:\mathcal T_i\to\mathcal T_i,
\]
and, for each \(i\in\mathcal I\), functions
\[
\kappa_i:\mathcal T_i\to\mathbb R_{++},
\qquad
\lambda_i:\mathcal T_i\to\mathbb R,
\]
such that, for every \(i\in\mathcal I\), every
\(t_i\in\mathcal T_i\), and every \(a\in\mathcal A\),
\[
U_i^{X'}[t_i,\alpha(a)]
=
\kappa_i(t_i)U_i^X[\tau_i(t_i),a]+\lambda_i(t_i).
\]
\end{corollary}

The affine relationship between these utilities means that two mechanisms are strategically analogous if and only if actions and types can be relabeled in a way that preserves agents' preferences over possibly mixed action profiles, that is, for every \(i\in\mathcal I\),
every \(t_i\in\mathcal T_i\), and every
\(A,A'\in\Delta(\mathcal A)\),
\[
U_i^X[\tau_i(t_i),A]\ge U_i^X[\tau_i(t_i),A']
\quad\Longleftrightarrow\quad
U_i^{X'}[t_i,\alpha_\# A]
\ge
U_i^{X'}[t_i,\alpha_\# A'].
\]
This characterization is useful for two reasons. First, it provides
a direct condition on relationships between payoffs that is easy to verify in applications.
Second, it shows that strategic analogy is a self-contained notion
of structural equivalence between mechanisms that can be used
independently of our epistemic framework.

To build intuition for strategic analogy, we provide a simple example.

\begin{example}\label{example:PPs}
An agent has value \(t_1\in\mathbb R\) for a good that may be allocated to her or remain unassigned; she may also be charged a payment for the good. This environment can be written as $\left(\mathcal I,\mathcal Y,\mathbb R^{\mathcal I},(u_i)_{i\in\mathcal I}\right),$ where $\mathcal I=\{1\},$ $\mathcal Y:=(\mathcal I\cup\{\varnothing\})\times \mathbb{R}_+,$ and $u_1\bigl(t_1,(y,p_1)\bigr)
=
t_1\mathbf{1}\{y=1\}-p_1.$

Consider a posted-price mechanism with price \(P\) in which the agent chooses whether to buy the good. If she buys it, she receives the good and pays \(P\). Otherwise, she gets and pays nothing. Formally, $X_P=\left(\mathcal A^P,\Phi^P\right),$ where
\[
\mathcal A^P=\mathcal A_1^P=\{\text{buy},\ \text{do not buy}\},
\qquad
\Phi^P(a_1)=
\begin{cases}
(1,P) & \text{if } a_1=\text{buy},\\
(\varnothing,0) & \text{if } a_1=\text{do not buy}.
\end{cases}
\] 
Then the family of posted-price mechanisms \(\{X_P:P\in\mathbb R_+\}\), where \(X_P=(\mathcal A^P,\Phi^P)\), forms a class of strategically analogous mechanisms. Indeed, for any \(P,P'\in\mathbb{R}_+\), let \(\alpha_1^{P\to P'}\) be the identity map and choose $\tau_1^{P'\to P}(t_1)=t_1+(P-P').$ Then, for all types \(t_1\in\mathbb R\), we have
\[
U_1^{X_{P'}}[t_1,\text{buy}]
= t_1-P'
= \bigl(t_1+(P-P')\bigr)-P
= U_1^{X_P}[\tau_1^{P'\to P}(t_1),\text{buy}],
\]
\[
U_1^{X_{P'}}[t_1,\text{do not buy}]
= 0
= U_1^{X_P}[\tau_1^{P'\to P}(t_1),\text{do not buy}].
\]
Thus, the condition of Corollary \ref{prop:strategicanalogy-affine} is satisfied with a scaling factor of $1$ and an additive constant of $0$. Intuitively, changing the posted price simply shifts the buyer's net value of buying from \(t_1-P'\) to \(t_1-P\). Thus, facing price \(P'\) as type \(t_1\) is strategically the same as facing price \(P\) as type \(t_1+(P-P')\).
\end{example}

The example illustrates how the remappings in Definition~\ref{def:strategicanalogy} should be interpreted. The action map \(\alpha\) identifies choices that play the same role in the two mechanisms, while the type map \(\tau\) identifies agents whose payoff tradeoffs are the same after those choices are matched. In the posted-price example, the action correspondence is trivial, since ``buy'' and ``do not buy'' have the same meaning at every price, while the type remapping adjusts the buyer's value to account for the change in price.

We now show how strategic analogy allows both equilibria and knowledge of equilibria to transfer across mechanisms.

\begin{theorem}\label{th:strategic-analogy-transfer}
Let \(X=(\mathcal A,\Phi)\) and \(X'=(\mathcal A',\Phi')\) be strategically analogous through \((\alpha,\tau)\), and fix a prior \(F\in\Delta(\mathcal T)\).
\begin{enumerate}
\item If \(\sigma\) is an equilibrium of \(X\) for prior \(\tau_\#F\), then the strategy profile defined by \(\sigma_i'(t_i):=(\alpha_i)_\#\sigma_i(\tau_i(t_i))\) is an equilibrium of \(X'\) for prior \(F\).
\item If, in addition, this equilibrium of \(X\) is commonly known and it is common knowledge that agents reason according to the analogy set \(\mathcal L_{\alpha,\tau}^{X,X'}\) from Definition~\ref{def:strategicanalogy}, then the equilibrium \(\sigma'\) of \(X'\) is commonly known.
\end{enumerate}
\end{theorem}
Here, \(\tau_\#F\) is the distribution over types obtained by relabeling each type through \(\tau\). Unlike strategic equivalence, which transfers equilibria under the same prior, strategic analogy generally requires the prior to change: an equilibrium of \(X'\) under \(F\) corresponds to an equilibrium of \(X\) under \(\tau_\#F\). This is because type \(t_i\) in \(X'\) takes the role of type \(\tau_i(t_i)\) in \(X\), so the distribution over types must be relabeled along with their strategies.

With these relabelings, equilibria transfer as in Theorem~\ref{th:strategic-equivalence-transfer}. Common knowledge of the original equilibrium similarly carries over if it is common knowledge that agents reason according to the relevant correspondences. Here, the analogy set is valid by Definition~\ref{def:strategicanalogy}, so again the substantive assumption in part~$(ii)$ is that agents recognize how actions and types relate across the two mechanisms and that it is common knowledge that they reason accordingly. The proof is in the appendix and uses the same reasoning closure to transfer the payoff comparisons that verify equilibrium.

Another consequence of the theorem is that, if an equilibrium of \(X\) is commonly known for \emph{every} prior, then, under the reasoning assumption in part~$(ii)$, an equilibrium of \(X'\) is commonly known for every prior as well. Here, knowledge of equilibria for all priors can be interpreted as understanding a general method for finding an equilibrium of \(X\) given a prior over types. The theorem then shows that agents who know such a method for \(X\) can also apply it to \(X'\), after appropriately relabeling the prior and the resulting strategies.

More broadly, strategic analogy can also be understood in abstraction from our epistemic framework, as a way of certifying a common strategic structure across mechanisms. Indeed, such shared structure may itself help agents transfer knowledge across mechanisms, even when their reasoning differs somewhat from our model.  In particular, the following proposition illustrates that strategic analogy preserves important strategic properties, such as the existence of dominant strategies.

\begin{definition}
A strategy \(s_i:\mathcal T_i\to \Delta(\mathcal A_i)\) is \textbf{dominant} for agent \(i\) in mechanism \(X\) if
\[
s_i(t_i)\in 
\argmax_{A_i\in \Delta(\mathcal A_i)}
U_i^X[t_i,A_i\otimes \mu_{-i}]
\qquad
\text{for every }t_i\in \mathcal T_i,\ \mu_{-i}\in \Delta(\mathcal A_{-i}).
\]
\end{definition}

\begin{proposition}\label{fact:dominant_exists_preserved}
Suppose \(X\) and \(X'\) are strategically analogous. Then \(X\)
admits dominant strategies for all agents if and only if \(X'\) does.
\end{proposition}

This perspective will be useful when discussing applications: when strategic analogy holds, we explain intuitively how agents' problems relate between mechanisms; when it fails, we identify which strategic features are present in one mechanism but not in the other.

\subsection{Application: single-item auctions}\label{sec:auctions_fpa}

This section applies strategic analogy to auction formats in the single-item environment from Example~\ref{example:fpa-dutch}. We first characterize which standard auction families form strategically analogous classes, and then turn to the implications for revenue-optimal auctions. Throughout, we focus on conventional highest-bid-wins formats:
\begin{definition}
A \textbf{standard auction} is a mechanism $X = (\mathcal A, \Phi)$ with the following properties:
\begin{enumerate}
\item \textbf{Common bid space.}
For every bidder \(i\), $\mathcal A_i=[\underline a,\infty)\cup\{\varnothing\}$ where \(\varnothing\) is a nonparticipation action after which a bidder never
receives the good and pays nothing.

\item \textbf{Highest bidder wins.}
Whenever at least one bidder participates, the object is allocated among the
highest participating bidders; ties are broken uniformly.

\item \textbf{Payment nonnegativity and symmetry.}
At every action profile, the vector of expected payments is nonnegative and symmetric under permutations of bidder labels.
\end{enumerate}
\end{definition}
This definition includes, among others, familiar formats such as first-price, second-price, and all-pay auctions, as well as intermediate formats in which the winner pays a weighted average of different ranked bids. We also introduce two specific formats---$k$th-price auctions with reserves and $k$th-price auctions with entry costs---that will be of interest throughout. 

\begin{definition}
    Fix $k\in\{1,\dots,N\}$. For any profile of bids, let $a^{(k)}$ denote the
$k$th-highest bid among participating bidders, with the convention that
$a^{(k)}=\underline a$ if fewer than $k$ bidders participate. Let $w(a)$ denote the winning
bidder.
\begin{enumerate}
    \item
    For $r\in\mathbb{R}_+$, the \emph{$k$th-price auction with reserve~$r$}
    is the mechanism $k\mathrm{PA}(r)=(\mathcal A^r,\Phi^r)$, where
    \[
    \mathcal A_i^r = \{\varnothing\}\cup [r,\infty),
    \qquad
    \Phi^r(a) =
    \begin{cases}
        (\varnothing,\;\mathbf 0)
            & \text{if } a_i=\varnothing \text{ for all } i, \\
        \bigl(w(a),\;
              a^{(k)}\, e_{w(a)}\bigr)
            & \text{otherwise.}
    \end{cases}
    \]
    \item
    For $c\in\mathbb{R}_+$, the \emph{$k$th-price auction with entry cost~$c$}
    is the mechanism $k\mathrm{EC}(c)=(\mathcal A^c,\Phi^c)$, where
    \[
    \mathcal A_i^c = \{\varnothing\}\cup \mathbb{R}_+,
    \qquad
    \Phi^c(a) =
    \begin{cases}
        (\varnothing,\;\mathbf 0)
            & \text{if } a_i=\varnothing \text{ for all } i, \\
        \bigl(w(a),\;
              a^{(k)}\,e_{w(a)}
              + c\displaystyle\sum_{i}\mathbf{1}\{a_i\neq\varnothing\}\;e_i
        \bigr)
            & \text{otherwise.}
    \end{cases}
    \]
    That is, every participant pays $c$, and the winner also pays her  $k$th-price payment.
\end{enumerate}
\end{definition}
Our first result characterizes the standard auctions that are strategically
analogous to a fixed \(k\)th-price auction with a reserve.

\begin{proposition}\label{prop:fpamaxclass}
A standard auction is strategically analogous to a $k$th-price auction with a reserve $r\in \mathbb{R}_+$ if and only if it is strategically equivalent to a $k$th-price auction with some reserve $r'\in \mathbb{R}_+$.
\end{proposition}

The result shows that \(k\)th-price auctions with reserves form a strategically analogous class. It also shows that this class is maximal in a natural sense: among standard auctions, any mechanism strategically analogous to a \(k\)th-price auction with a reserve is, up to relabeling of actions, itself a \(k\)th-price auction with some reserve. Thus, if a designer wants to adjust a \(k\)th-price auction while keeping it standard, the reserve price is effectively the only parameter she can vary. Indeed, the practice of fixing the auction format and varying only the reserve is widespread in industry. Sotheby's, Christie's, and eBay all commit to a standard auction format and use the reserve as their primary tuning parameter across sales.\footnote{\url{https://www.sothebys.com/en/buy-sell} \\ \url{https://www.christies.com/en/help/auction-help-library/how-do-auction-houses-work} \\ \url{https://www.ebay.com/help/selling/listings/selling-auctions/reserve-prices?id=4143}.} In online advertising, major ad exchanges use the first-price auction as their standard format, with publishers optimizing revenue by adjusting reserve prices across impressions and keywords.\footnote{See e.g.\ \url{https://support.google.com/admanager/answer/9298211}.}

While the proof is relegated to the appendix, we provide an intuition for why reserve prices preserve strategic analogy. Fix two reserves $r,r'\in\mathbb{R}_+$ and consider the bijections
\[
\alpha_i^{r\to r'}(a)=
\begin{cases}
    \varnothing & \text{if } a=\varnothing,\\
    a+(r'-r)   & \text{if } a\in [r,\infty),
\end{cases}
\qquad
\tau_i^{r'\to r}(t_i)=t_i-(r'-r).
\]
These maps shift every active bid by $r'-r$ and every value by the opposite amount. Since $\alpha^{r\to r'}$ preserves participation, bid rankings, and ties, it preserves each bidder's probability of winning. If at least $k$ bidders participate, the $k$th-highest bid changes by $r'-r$; otherwise, any winner pays the reserve, which changes by the same amount. Thus, for every bidder $i$, type $t_i$, and action profile $a$,
\begin{align*}
U_i^{k\mathrm{PA}(r')}[t_i,\alpha^{r\to r'}(a)]
&=q_i(a)\bigl(t_i-a^{(k)}-(r'-r)\bigr)\\
&=q_i(a)\bigl(\tau_i^{r'\to r}(t_i)-a^{(k)}\bigr)
=U_i^{k\mathrm{PA}(r)}[\tau_i^{r'\to r}(t_i),a],
\end{align*}
where $q_i(a)$ is bidder $i$'s common probability of winning in the two corresponding profiles. The condition of Corollary~\ref{prop:strategicanalogy-affine} is therefore satisfied with a scaling factor of $1$ and an additive constant of $0$. Note that this transformation is analogous to the case of posted-price mechanisms from Example~\ref{example:PPs}. Intuitively, when the reserve rises from $r$ to $r'$, an agent can think of her choice in the new auction as the choice of a lower type in the old one.

\begin{remark}\label{rem:payment_rule}
One might wonder whether the strategic analogy of auctions with reserves is special to the \(k\)th-price family. To answer this question, note first that the above construction of bijections $\alpha$ and $\tau$ works thanks to three properties of the $k$th-price format. First, only the winner pays; second, translating actions by the same constant preserves the winner; third, the winner's payment under reserve $r$, written as $\pi^r(a)$, satisfies:
\[
\pi^r(a)=\pi^0(a-r\mathbf 1)+r.
\]
As a result, changing the reserve can be absorbed by translating both bids and values. The action remapping \(a_i\mapsto a_i-r\) preserves the allocation rule, while the type remapping \(t_i\mapsto t_i-r\) offsets the payoff consequence of the additional \(r\) paid by the winner. The same argument can therefore be applied to other formats with these properties, like the \( 1.5\)-price auction, in which the winner pays an average of the first- and second-highest bids. Thus, \(1.5\)-price auctions with different reserves are also strategically analogous.
\end{remark}

The maximality statement in Proposition~\ref{prop:fpamaxclass} also tells us which families of auctions are \emph{not} strategically analogous:
\begin{corollary}\label{cor:different_k}
For $k\neq k'$, no $k$th-price auction with a reserve is strategically
analogous to any $k'$th-price auction with a
reserve.
\end{corollary}
For $k=1$ and $k'=2$ this also follows from
Proposition~\ref{fact:dominant_exists_preserved}, since the second-price
auction is solvable in dominant strategies while the first-price auction
is not; the corollary extends this observation to any $k\neq k'$. We also get the following result:
\begin{proposition}\label{prop:entry-costs}
The family of \(k\)th-price auctions with
strictly positive entry costs, $\{k\mathrm{EC}(c):c>0\},$ is a class of strategically analogous mechanisms. However, if the zero-entry-cost
auction is included, then the resulting family of \(k\)th-price auctions with
weakly positive entry costs, $\{k\mathrm{EC}(c):c\ge 0\},$ is not a class of strategically analogous mechanisms.
\end{proposition}

We first discuss the negative result. To see why it holds, note that a $k$th-price auction with an entry cost of zero is also a $k$th-price auction with a \emph{reserve} of zero, and therefore Proposition \ref{prop:fpamaxclass} applies to it. Since $k$th-price auctions with entry costs $c>0$ are standard but not strategically equivalent to $k$th-price auctions with reserves, they cannot be strategically analogous to the one with $c=0$.

Consequently, $k$th-price auctions with reserves are more alike in their
strategic structure than $k$th-price auctions with entry costs. This
observation may not be obvious a priori. After all, both families contain
the standard $k$th-price auction as a special case. It turns out, however, that 
introducing a strictly positive reserve to this baseline format changes the underlying strategic problem less than introducing a strictly positive entry cost. To get an intuition for this difference, consider how the agent's decision about whether to enter the auction changes as these parameters vary. 

First, note that the entry decision is just as simple with a reserve as it is without one: when the reserve is $r$, nonparticipation weakly dominates every active bid for types $t_i<r$, while bidding $r$ weakly dominates nonparticipation for types $t_i>r$. A strictly positive entry cost, by contrast, makes the participation decision strategically complex. Since
an agent who enters and loses still pays~$c$, entry is worthwhile only if
the bidder expects a sufficiently high chance of winning, which in turn means that the entry decision depends on her beliefs about how likely other bidders are to enter and how aggressively they are likely to bid. In this sense, making the entry cost strictly positive introduces a qualitatively new strategic consideration, which breaks analogy to the baseline case. 

Nevertheless, $k$th-price auctions with \emph{strictly positive} entry costs are strategically analogous to one another. Although belief-dependent entry arguably makes auctions with entry costs \(c>0\) harder than reserve-price auctions, this underscores that strategic analogy concerns the \emph{structural similarity} of strategic problems, not their \emph{difficulty}---auctions with strictly positive entry costs may be harder to play, but the reasoning needed to play in one transfers to any other.

\subsubsection{Revenue-maximizing auctions} One might ask whether a designer facing many possible bidder distributions can restrict attention to a single strategically analogous class of auctions without sacrificing revenue. We therefore consider a revenue-maximizing designer facing symmetric IPV priors \(F\), who may select a different mechanism for each prior. In each case, the mechanism must allow agents not to participate, that is, contain an action which, when taken, does not charge the agent and does not give her the good. For a given prior, we call an auction \emph{optimal} if it admits an equilibrium attaining the highest expected revenue among such mechanisms. We consider priors \(F\) that are absolutely continuous, have support \([\underline v,\overline v]\) with \(0\leq\underline v<\overline v<\infty\), and have a positive density almost everywhere on \((\underline v,\overline v)\).  

If we restrict attention to regular priors in the sense of \cite{myerson1981optimal}, there does indeed exist a single strategically analogous class containing an optimal auction for each such prior.\footnote{A distribution \(F\) with density \(f\) is regular if its virtual value \(\varphi_F(v)=v-\frac{1-F(v)}{f(v)}\) is weakly increasing on its support.}
In fact, a second-price auction with an appropriately chosen common reserve is optimal, and these auctions form a strategically analogous class by Proposition~\ref{prop:fpamaxclass}. Once irregular symmetric priors are allowed, however, this conclusion no longer holds under mild restrictions on the class of auctions. To state the result, we first introduce a technical condition on payments.

\begin{definition}
An auction is \textbf{payment-regular} if, for every bidder \(i\),
it satisfies the following properties:
\begin{enumerate}
    \item \textbf{Own-bid continuity.}
    Let \(\pi_w(a)\) and \(\pi_\ell(a)\) denote bidder \(i\)'s
    expected payments conditional on winning and losing. For every \(a_{-i}\), the functions
    \(\pi_w(\cdot,a_{-i})\) and \(\pi_\ell(\cdot,a_{-i})\) are
    continuous on the sets of bids with which bidder \(i\) can
    win and lose, respectively.

\item \textbf{Joint local boundedness.}
For every fixed set of participating bidders, bidder \(i\)'s expected
payment \(p_i(a)\) is locally bounded as a function of the participating
bids.
\end{enumerate}
\end{definition}

These conditions hold for all the auction formats discussed in this paper, including \(k\)th-price auctions with reserves and all-pay auctions, and are preserved by adding a fixed entry fee.

\begin{proposition}\label{prop:no-class-all-myerson}
Let \(\mathcal C\) be a strategically analogous class of symmetric mechanisms, each offering every bidder a nonparticipation action. If \(\mathcal C\) contains a standard, payment-regular auction, then it cannot contain mechanisms attaining optimal revenue in symmetric equilibrium for every symmetric IPV prior.
\end{proposition}

To understand the result, recall from \cite{myerson1981optimal} that for irregular priors, the symmetric revenue-maximizing auction sometimes allocates the object at random among a subset of bidders. Thus, a strategically analogous class that implements revenue-optimal auctions for all
symmetric IPV priors would have to contain mechanisms whose equilibria exhibit
such randomization, as well as mechanisms whose equilibria assign the object to the
highest bidder above a reserve. Proposition~\ref{prop:no-class-all-myerson} shows that these two allocation structures are generally incompatible within any strategically analogous class that contains a well-behaved highest-bid-wins auction. Indeed, this negative result provides a reason why a designer may avoid such randomized auction formats in practice even when she knows that she faces an irregular prior: bidders familiar with standard auctions may lack the strategic understanding needed to play an equilibrium involving randomization, making such formats less attractive despite theoretically higher expected revenue.

The requirement that \(\mathcal C\) contain a standard, payment-regular auction serves to rule out pathological constructions: without some restriction of this kind, one could build mechanisms whose message spaces are engineered to select arbitrary equilibrium allocations. Note also that the result does not require every mechanism in \(\mathcal C\) to be a standard auction, but only requires that \(\mathcal C\) contain one such auction. Thus, the impossibility arises as soon as the designer sometimes wants to use a more conventional auction format.

It is also worth noting that the positive result does not generally extend to asymmetric regular priors. When bidders draw values from different regular distributions \(F_1,\dots,F_N\), the optimal auction ranks them by \emph{virtual values} rather than by values alone. This typically requires bidder-specific reserves, which can break the translation argument used here. To illustrate, consider two bidders and two first-price auctions $\mathrm{FPA}(r_1,r_2)$ and $\mathrm{FPA}(r_1',r_2')$ with $r_1' - r_1 \neq r_2' - r_2.$ The natural candidate bijections are the bidder-specific analogues of the translations used in Proposition~\ref{prop:fpamaxclass}:
\[
\alpha_i(b_i) = b_i + (r_i' - r_i), \qquad \tau_i(t_i) = t_i - (r_i' - r_i).
\]
These preserve each bidder's payoff conditional on winning, just as in the common-reserve case. However, when the shifts $r_i' - r_i$ differ across bidders, they do not preserve who actually wins. Indeed, this would require:
\[
b_1 > b_2 \quad \Longleftrightarrow \quad b_1 + (r_1' - r_1) > b_2 + (r_2' - r_2)
\]
for all feasible bid pairs, which fails whenever $r_1' - r_1 \neq r_2' - r_2$.

\subsection{Application: procurement auctions}\label{sec:auctions_procurement}
We now apply our framework to a procurement setting. A buyer awards a contract to one of \(N\) sellers based on a score over price and quality. We consider two common price-scoring rules---a linear rule and a ratio rule---and show that changing the weight on price preserves strategic analogy under the linear rule but not the ratio rule.\footnote{See \citet{dini2006scoring} for a discussion of linear and lowest-bid ratio scoring rules in procurement.}

\paragraph{Environment.} A buyer awards a contract to one of \(N\) sellers. Each seller \(i\) has a private type \(t_i = c_i \in \mathcal T_i = \mathbb R\) specifying her cost. An outcome \(y=(w,p)\in\mathcal Y=\{1,\dots,N\}\times\mathbb R\) specifies the winning seller and the payment. Seller \(i\)'s utility is
\[
u_i\bigl(c_i,(w,p)\bigr)
=
\begin{cases}
    p - c_i & \text{if } w = i,\\
    0 & \text{if } w \neq i.
\end{cases}
\]
\paragraph{First-score auctions.} In the mechanisms we study, each seller submits a price \(b_i\). In addition to price, each seller \(i\) has a fixed and publicly known quality score \(q_i > 0\) which reflects observable characteristics such as past performance or delivery capacity. Given a weight \(\lambda\in(0,1)\) and a price-scoring rule \(S_p\), seller \(i\)'s total score is
\[
S_\lambda(b)_i = \lambda\,S_p(b_i,b_{-i}) + (1-\lambda)\, q_i.
\]
The contract is awarded to one of the highest-scoring sellers, according to a common fixed tie-breaking rule. The winner is paid the price she submitted. Formally, letting $w(b)$ denote the seller selected as the winner from among the highest scorers, the outcome rule is $\Phi(b) = \bigl(w(b),\, b_{w(b)}\bigr).$ 

To complete the specification of a first-score auction, it remains to choose the price-scoring rule \(S_p\). We consider two such rules, each giving rise to a family of mechanisms indexed by the weight \(\lambda\):
\begin{enumerate}
    \item The \emph{linear-score} auction \(L(\lambda)\), where \(\mathcal A_i^{L(\lambda)}=\mathbb R\) and \(S_\lambda(b)_i=-\lambda b_i+(1-\lambda)\, q_i\).
    \item The \emph{ratio-score} auction \(R(\lambda)\), where \(\mathcal A_i^{R(\lambda)}=\mathbb{R}_{++}\) and \(S_\lambda(b)_i=\lambda\frac{\min_j b_j}{b_i}+(1-\lambda)\, q_i\).
\end{enumerate}
We ask whether varying \(\lambda\) preserves strategic analogy within each family. As it turns out, the linear family forms a class of strategically analogous mechanisms, while the ratio family does not.

\begin{proposition}\label{prop:procurement-scores}
The linear family \(\{L(\lambda):\lambda\in(0,1)\}\) is a class of strategically analogous mechanisms. However, unless the fixed quality scores \(q_i\) are all
equal, the ratio family \(\{R(\lambda):\lambda\in(0,1)\}\) is not a class of
strategically analogous mechanisms.
\end{proposition}

Intuitively, the linear family is strategically analogous because changing \(\lambda\) can be absorbed by seller-specific translations of bids and costs. After the bid relabeling, all sellers' scores end up rescaled by the same positive constant, leaving the winner's identity unchanged. The corresponding change in the winner's payment is then exactly offset by the cost relabeling.

To see why the ratio family fails, notice that when \(\lambda\) is small, a sufficiently low-quality seller can never win: the price score is bounded above by \(\lambda\), which in certain cases cannot outweigh other participants' quality bonuses. As \(\lambda\) increases, however, the set of bidders who have a chance of winning changes. But a mechanism in which some seller never wins cannot be strategically analogous to one in which that seller wins for some bid profiles, since the affine payoff condition would force her payoff in the latter to be constant across all bid profiles, which it is not.

More generally, in the linear family, both a seller's profit from winning and her score, which determines whether she wins, depend directly on her bid. A change in $\lambda$ can therefore be offset by shifting each seller's bid and cost by the same seller-specific amount. In the ratio family, by contrast, a seller's profit from winning depends directly on her bid, while her price score depends on its value relative to the lowest bid. Adjusting bids to preserve score rankings therefore need not preserve profits in the same way, meaning that changing $\lambda$ can qualitatively alter the seller's bidding problem.

\subsection{Application: input versus output pricing}\label{sec:pricing-analogy}

In the last application, we consider a provider offering a service using limited processing capacity whose effectiveness varies over time. This setting captures a range of industries. A landowner leases acres whose yield depends on rainfall and soil conditions; a shipping line sells container slots whose value to shippers depends on routing and port congestion; a consulting firm sells associate time whose productivity depends on relevant institutional knowledge.

We compare two ways of pricing such services. Under \emph{output pricing}, the provider charges based on the realized outcome, such as crop yield or the date on which cargo is delivered. Under \emph{input pricing}, the provider charges for the resources the customer uses---per acre, per container slot, or per billable hour---with a tariff that does not vary with the current productivity of those resources. We show that input pricing, combined with proportional rationing of overdemand, generates a class of strategically analogous mechanisms as provider efficacy varies. By contrast, output pricing generically does not.

Such input-based pricing is common in a range of settings. In community solar programs, for example, subscribers often pay for shares of installed capacity, even though the electricity generated by that capacity varies with environmental conditions.\footnote{\url{https://docs.nlr.gov/docs/fy23osti/86242.pdf}} Professional-services firms often charge per billable hour of staff time rather than based on concrete deliverables. A further example comes from U.S.\ agricultural land markets: the most common arrangement between a landowner and a tenant farmer is a \emph{cash rent} lease, under which the tenant pays based on acreage regardless of how much the land produces in that season.\footnote{\url{https://www.canr.msu.edu/news/farm_land_rental_agreements_and_arrangements}} An alternative lease model, called a \emph{crop share} lease, corresponds to output pricing. There, the landlord receives a fraction of the realized harvest, so payments vary directly with productivity.

\paragraph{Environment.}
A service provider has a resource capacity (normalized to \(1\)) that she can allocate across \(N\geq 2\) buyers. This capacity has efficacy $e\in \mathbb{R}_{++}$, that is, when capacity $z_i$ is devoted to serving agent $i$, she receives the output $x_i = e \, z_i$. Formally, the environment is $\left(\mathcal I,\mathcal Y,\mathcal T,\mathcal U\right),$ where 
\[
\mathcal I=\{1,\dots,N\},
\qquad
\mathcal T_i=\mathcal V,
\qquad
\mathcal Y=\mathbb R_+^{\mathcal I}\times\mathbb R^{\mathcal I},
\]
and \(\mathcal V\) is the set of all weakly increasing, weakly concave
functions \(v:\mathbb R_+\to\mathbb R\) satisfying \(v(0)=0\). The agent's type is therefore a \emph{function} specifying her value for different levels of delivered output. An outcome
\(y=(x,p)\in\mathcal Y\) specifies the amount of output \(x_i\) allocated to
each buyer and the payment \(p_i\) made by each buyer. Each buyer's Bernoulli
utility function is
\[
u_i(v_i,y)=v_i(x_i)-p_i.
\]
We now introduce two families of mechanisms for this environment. Both will be defined in terms of a \emph{regular price scheme}: a function \(P:\mathbb R_+\to\mathbb R_+\) that is continuous, strictly increasing, strictly convex, and satisfies \(P(0)=0\).

\paragraph{Input-pricing mechanisms.}
Fix a regular price scheme \(P\); a buyer who uses capacity \(z_i\) pays \(P(z_i)\), regardless of the current efficacy. For each efficacy level \(e>0\), define the mechanism
\(X^{\mathrm{in}}_e=([0,1]^{\mathcal I},\Phi^{\mathrm{in}}_e)\),
where each buyer chooses a requested capacity $a_j\in [0,1]$. If orders exceed the provider's capacity, it is rationed proportionally; thus, given an action profile \(a\), buyer \(i\) is allocated capacity
\[
z_i(a):=
\begin{cases}
a_i\,\min\left\{1,\,1\big/\textstyle\sum_{j\in\mathcal I}a_j\right\},
& \text{if } \sum_{j\in\mathcal I}a_j>0,\\
0, & \text{otherwise}.
\end{cases}
\]
The outcome rule is $\Phi^{\mathrm{in}}_e(a)
=
\bigl(x(a,e),\,p(a)\bigr),$ where $x_i(a,e)=e\,z_i(a)$ and $p_i(a)=P\bigl(z_i(a)\bigr).$ For every buyer \(i\), type \(v_i\in\mathcal T_i\), and pure action profile
\(a\in[0,1]^{\mathcal I}\),
\[
U_i^{X^{\mathrm{in}}_e}[v_i,a]
=
v_i\bigl(e\,z_i(a)\bigr)-P\bigl(z_i(a)\bigr).
\]

\paragraph{Output-pricing mechanisms.}
The tariff is defined over units of output rather than over capacity. Fix a regular price scheme \(P\). For each efficacy level \(e>0\), define the mechanism
\(X^{\mathrm{out}}_e=([0,e]^{\mathcal I},\Phi^{\mathrm{out}}_e)\),
where each buyer requests output $a_i\in [0,e]$. Since capacity may still need to be rationed, given an action profile \(a\), buyer \(i\)
receives output
\[
x_i(a,e):=
\begin{cases}
a_i\,\min\left\{1,\,e\big/\textstyle\sum_{j\in\mathcal I}a_j\right\},
& \text{if } \sum_{j\in\mathcal I}a_j>0,\\
0, & \text{otherwise}.
\end{cases}
\]
The outcome rule is $\Phi^{\mathrm{out}}_e(a)
=
\bigl(x(a,e),\,p(a,e)\bigr),$ where $p_i(a,e)=P\bigl(x_i(a,e)\bigr).$ Thus, for every buyer \(i\), type \(v_i\in\mathcal T_i\), and pure action profile
\(a\in[0,e]^{\mathcal I}\),
\[
U_i^{X^{\mathrm{out}}_e}[v_i,a]
=
v_i\bigl(x_i(a,e)\bigr)-P\bigl(x_i(a,e)\bigr).
\]

\begin{proposition}\label{prop:pricing-analogy}
For any regular price scheme \(P\), the family of input-pricing mechanisms
\(\{X^{\mathrm{in}}_e:e>0\}\) is a class of strategically analogous
mechanisms. The family of output-pricing mechanisms
\(\{X^{\mathrm{out}}_e:e>0\}\) is a class of strategically analogous
mechanisms if and only if \(P\) is isoelastic, that is,
\(P(x)=Ax^\rho\) for every \(x\geq0\), for some \(A>0\) and \(\rho>1\).
\end{proposition}

The proposition shows that the conditions a price scheme must satisfy
for the output-pricing family to be strategically analogous are very
stringent. Thus, for a generic regular
price scheme, output pricing does not preserve strategic analogy as
efficacy varies.

This suggests a sense in which input-based tariffs may be simpler for repeated customers. Over time, such customers may face different levels of need for the service (which in the model corresponds to different value functions over output) and thus know how to optimally choose orders under different valuations. Input-based tariffs make this knowledge easier to reuse when the effectiveness of the provider's capacity changes. A customer who understands how to choose capacity for different values of output can then reinterpret a change in effectiveness as a change in her effective valuation, and adjust her order accordingly.

\section{Discussion}\label{sec:discussion}

We develop a theory of how strategic reasoning transfers across mechanisms and formalize a sense in which genuinely different mechanisms can nevertheless be \emph{strategically analogous}. Our characterizations of strategically analogous classes in different settings yield falsifiable predictions about when agents can transfer strategic understanding across mechanisms.

The notion of strategic analogy could be extended in several directions. One could, for instance, consider mechanisms that are ``almost'' strategically analogous but differ in having additional actions that are dominated or redundant. Our current definition requires full bijections between action spaces, reflecting the baseline assumption that agents cannot easily identify such actions. One could relax this by first endowing agents with the ability to recognize dominated or redundant actions, having them eliminate these, and then requiring bijections only between the remaining ones. More substantially, one could extend the notion of analogy to compare mechanisms with different numbers of agents, or define one-directional reductions in which simpler mechanisms could be mapped to restricted versions of more complex ones. Such reductions would provide a natural notion of relative strategic complexity across mechanisms.

Another direction for future work concerns how analogies between mechanisms are communicated to agents. In the baseline model, we have assumed that the designer can publicly explain these relationships to agents and have them internalize the relevant analogy set. We have, however, remained agnostic about how these relationships are communicated. We explore one possibility in Appendix~\ref{app:rules_based}, where we assume that the designer explains to agents how the \emph{rules} of analogous mechanisms relate. Agents then use these rule mappings to infer the corresponding payoff analogies. This mode of communication is particularly natural, as it concerns concrete features of the mechanism that must be communicated anyway. However, as we show in that section, not all strategic analogies can be conveyed in this manner. This suggests a sense in which some analogies are easier to explain than others, which could be explored in future work. Indeed, since an analogy's usefulness depends on how easily it can be explained, accounting for ease of communication could provide a natural refinement of strategic analogy.

Finally, the framework we propose could be used to study strategic reasoning beyond the transfer of equilibrium knowledge. Our analysis endows agents with the ability to reason through payoff analogies across mechanisms. One could, however, endow them with the ability to recognize structural properties of a \emph{single} mechanism, such as symmetry or single-crossing conditions, and ask which strategic problems can be solved using those tools alone. This could provide a way of studying the difficulty of mechanisms outside dominant-strategy settings.

\appendix
\titleformat{\subsection}[runin]
  {\normalfont\bfseries}
  {\thesubsection.}
  {0.5em}
  {}
  [.]
\titlespacing*{\subsection}
  {0pt}
  {1ex plus .2ex minus .1ex}
  {0.75em}

\section{Rules-based transformations}\label{app:rules_based}
In the baseline model, we have assumed that the designer can communicate analogies between payoff situations to agents. We have, however, remained agnostic about how this communication happens; in this section, we provide one possible model of it. Intuitively, we consider a language in which the designer communicates correspondences between mechanisms' \emph{rules}. We discussed this idea informally in Section~\ref{sec:strategic-equivalence}, where we argued that strategic equivalence can be communicated by explaining how what happens after an action profile in one mechanism corresponds to what happens after an action profile in the other. Here, we formalize this idea and extend it to strategic analogy. In particular, we show how agents can use correspondences between mechanisms' rules, together with their preferences over outcomes, to infer the mappings $\alpha$ and $\tau$ between strategic problems. In these cases, the designer need not communicate type correspondences directly: it suffices to describe how the mechanisms' actions and outcomes relate. 

For each mechanism \(X\), let
\[
\mathcal Y_X:=\bigcup_{a\in\mathcal A}\operatorname{supp}\Phi(a)
\]
denote its set of attainable outcomes. We then define \emph{rules-analogous} mechanisms as follows:

\begin{definition}\label{def:rules-analogy}
Let \(X=(\mathcal A,\Phi)\) and \(X'=(\mathcal A',\Phi')\) be mechanisms for the same environment. The mechanisms \(X\) and \(X'\) are \textbf{rules-analogous} if there exist bijections
\[
\alpha^r=(\alpha_i^r)_{i\in\mathcal I},
\quad
\alpha_i^r:\mathcal A_i\to\mathcal A_i',
 \qquad
\gamma^r:\mathcal Y_X\to\mathcal Y_{X'},
\]
such that, for every pure action profile \(a\in\mathcal A\),
\begin{equation}\label{eq:rulemap}
\gamma^r_\#\bigl(\Phi(a)\bigr)
=
\Phi'\bigl(\alpha^r(a)\bigr).
\end{equation}
A collection of mechanisms is a \textbf{class of rules-analogous mechanisms} if every pair of mechanisms in the collection is rules-analogous.
\end{definition}
Intuitively, an action profile in one mechanism generates the same distribution of outcomes as the relabeled action profile in the other, after relabeling outcomes. Importantly, this relation concerns only the mechanisms' \emph{rules} and does not depend on agents' preferences.

\begin{example}[First-price auctions with reserves]
Consider the auction environment from the main text and two first-price auctions with reserves \(r\) and \(r'\). A rules analogy between them is given by
\begin{equation}\label{eq:rules-reserve-maps}
\alpha_i^r(a_i)
=
\begin{cases}
\varnothing & \text{if }a_i=\varnothing,\\
a_i+(r'-r) & \text{otherwise},
\end{cases}
\qquad
\gamma^r(w,p)
=
\left(w,\bigl(p_i+(r'-r)\mathbf 1\{w=i\}\bigr)_{i\in\mathcal I}\right),
\end{equation}
where \(w\) is the winner and \(p\) is the vector of payments.

Intuitively, shifting the reserve and all active bids by the same amount leaves the allocation unchanged and shifts the winner's payment by that amount. Agents can then infer that this payment change is offset by an equal change in their values.
\end{example}

In this example, the outcome transformation also explains the type transformation used to establish strategic analogy. Bidding \(a_i+(r'-r)\) as type \(t_i\) in the auction with reserve \(r'\) gives the same payoff as bidding \(a_i\) as type \(t_i-(r'-r)\) in the auction with reserve \(r\), provided the other bids are shifted accordingly. Thus, the change in payment can be absorbed by a change in the bidder's value. More generally, a rules analogy induces a payoff analogy whenever a type relabeling offsets the change in outcomes. The following definition formalizes how a rules analogy can induce a payoff analogy.

\begin{definition}\label{def:rules-induced-analogy}
Suppose \(X\) and \(X'\) are rules-analogous under \(\alpha^r,\gamma^r\). Let
\[
\tau^r=(\tau_i^r)_{i\in\mathcal I},
\quad
\tau_i^r:\mathcal T_i\to\mathcal T_i,
\]
be a profile of bijections such that, for every \(i\in\mathcal I\), \(t_i\in\mathcal T_i\), and \(y\in\mathcal Y_X\),
\begin{equation}\label{eq:rules-payoff}
u_i\bigl(t_i,\gamma^r(y)\bigr)
=
u_i\bigl(\tau_i^r(t_i),y\bigr).
\end{equation}
Then the analogy set
\[
\mathcal L_{\alpha^r,\tau^r}^{X,X'}
:=
\Bigl\{
(X,\tau_i^r(t_i),a)
\;\indiffu{i}\;
(X',t_i,\alpha^r(a))
:
i\in\mathcal I,\ t_i\in\mathcal T_i,\ a\in\mathcal A
\Bigr\}
\]
is a \textbf{payoff analogy induced by the rules analogy}.
\end{definition}

Any such analogy set is valid. Indeed, combining \eqref{eq:rulemap} and \eqref{eq:rules-payoff} gives
\[
U_i^{X'}[t_i,\alpha^r(a)]
=
\mathbb E_{y\sim\Phi(a)}
\bigl[u_i(t_i,\gamma^r(y))\bigr]
=
\mathbb E_{y\sim\Phi(a)}
\bigl[u_i(\tau_i^r(t_i),y)\bigr]
=
U_i^X[\tau_i^r(t_i),a],
\]
and so Corollary~\ref{prop:strategicanalogy-affine} applies with a scaling factor of \(1\) and an additive constant of \(0\).

We next explore which auction formats studied in the paper are rules-analogous and whether these rules analogies induce payoff analogies.

\begin{proposition}\label{prop:functionalsim_examples}
Each of the following is a class of rules-analogous mechanisms:
\begin{enumerate}
    \item First-price auctions with reserves, \(\{1\mathrm{PA}(r):r\geq0\}\).
    \item First-price auctions with strictly positive entry costs, \(\{1\mathrm{EC}(c):c>0\}\).
    \item Linear-score auctions, \(\{L(\lambda):\lambda\in(0,1)\}\).
\end{enumerate}
For first-price auctions with reserves and linear-score auctions, the rules analogies constructed below induce the payoff analogies establishing strategic analogy in the main text. By contrast, for first-price auctions with distinct strictly positive entry costs, the constructed rules analogies do not induce payoff analogies whenever a bidder can enter and lose.
\end{proposition}

\begin{proof}
In each case, we provide action and outcome bijections satisfying \eqref{eq:rulemap}. For first-price auctions with reserves and linear-score auctions, we also provide type bijections satisfying \eqref{eq:rules-payoff}. We use the outcome spaces and common tie-breaking rules specified in the main text.

\emph{First-price auctions with reserves.}
Fix reserves \(r,r'\geq0\), and use the maps in \eqref{eq:rules-reserve-maps}. The action map preserves participation, bid rankings, and ties, and changes the winner's payment by \(r'-r\). Both maps are bijections, with inverses obtained by interchanging \(r\) and \(r'\). Thus \eqref{eq:rulemap} holds.

Let \(\tau_i^r(t_i)=t_i-(r'-r)\). For every attainable outcome \((w,p)\),
\[
u_i\bigl(t_i,\gamma^r(w,p)\bigr)
=
\bigl(t_i-(r'-r)\bigr)\mathbf 1\{w=i\}-p_i
=
u_i\bigl(\tau_i^r(t_i),(w,p)\bigr).
\]
Hence \eqref{eq:rules-payoff} holds.

\emph{First-price auctions with entry costs.}
Fix \(c,c'>0\), and take \(\alpha_i^r=\mathrm{id}\) for every bidder \(i\). Define
\[
\gamma^r(w,p)
=
\left(
w,\,
\bigl(p_i+(c'-c)\mathbf 1\{p_i>0\}\bigr)_{i\in\mathcal I}
\right).
\]
Since entry costs are strictly positive and bids are nonnegative, bidder \(i\) participates if and only if \(p_i>0\). The outcome map therefore changes each participant's entry payment from \(c\) to \(c'\), leaving the allocation and the winner's bid payment unchanged. Interchanging \(c\) and \(c'\) gives the inverse, so \eqref{eq:rulemap} holds.

If \(c\neq c'\) and bidder \(i\) enters and loses, her payoff is \(-c\) under the original outcome and \(-c'\) under its image, regardless of her value. Thus no type bijection can satisfy \eqref{eq:rules-payoff}.

\emph{Linear-score auctions.}
Fix weights \(\lambda,\lambda'\in(0,1)\), and write \(d_i:=(1/\lambda'-1/\lambda)q_i\). Define \(\alpha_i^r(b_i)=b_i+d_i\) and \(\gamma^r(w,p)=(w,p+d_w)\), where \(p\) is the payment to the winning seller. Both maps are bijections. Moreover,
\[
-\lambda'(b_i+d_i)+(1-\lambda')q_i
=
\frac{\lambda'}{\lambda}
\bigl[-\lambda b_i+(1-\lambda)q_i\bigr].
\]
Thus all scores are multiplied by the same positive constant, preserving the highest-scoring sellers and the tie-breaking probabilities. The winner's payment changes by \(d_w\), as prescribed by \(\gamma^r\). Hence \eqref{eq:rulemap} holds.

Let \(\tau_i^r(c_i)=c_i-d_i\). For every attainable outcome \((w,p)\),
\[
u_i\bigl(c_i,\gamma^r(w,p)\bigr)
=
\mathbf 1\{w=i\}(p+d_i-c_i)
=
u_i\bigl(\tau_i^r(c_i),(w,p)\bigr).
\]
Hence \eqref{eq:rules-payoff} holds. In both the reserve and linear-score cases, the action and type maps are those used to establish strategic analogy in the main text.
\end{proof}
One might wonder why these rules analogies induce payoff analogies for reserve-price and linear-score auctions, but not for entry-cost auctions. To see this, note that in the former two cases, the rules transformations change only the winner's payment. As a result, these changes can be offset by corresponding changes in bidders' values or sellers' costs. By contrast, changing the entry cost changes the payment of \emph{every} participant, including those who lose. Since a losing bidder's payoff does not depend on her value, no type adjustment can offset this payment change under the specified rules analogy. The above result therefore suggests that certain analogy sets may be easier to communicate than others; we leave this question for future work.

\section{Omitted proofs}

\subsection{Proof of Proposition \ref{th1:eqchar}}

\((\Leftarrow)\). Suppose the affine condition holds and define \(\mathcal L^*\) as the set of all payoff analogies $(X,t_j,A)\indiffu{j}(X',t_j',A')$ such that
\[
U_j^{X'}[t_j',A']
=
\ell^j_{(X,t_j)\to (X',t_j')}
\bigl(U_j^X[t_j,A]\bigr).
\]
By assumption, \(\mathcal L\subseteq \mathcal L^*\). We first show that
\(\mathcal L^*\) is closed under the symmetry and mixture operations used to
define \(\overline{\mathcal L}\). It will follow that
\(\overline{\mathcal L}\subseteq \mathcal L^*\). We then show that
\(\mathfrak R^{\mathrm{true}}\) is closed under \(\mathcal L\). First, consider the transfer step. Suppose $\bigl(U_j^X[t_j,A^1]\ge U_j^X[t_j,A^2]\bigr)
\in \mathfrak R^{\mathrm{true}}$ and
\[
\bigl\{
(X,t_j,A^1)\indiffu{j}(X',t_j',A^3),\;
(X,t_j,A^2)\indiffu{j}(X',t_j',A^4)
\bigr\}
\subseteq \mathcal L^*.
\]
Then
\[
U_j^{X'}[t_j',A^3]
=
\ell^j_{(X,t_j)\to (X',t_j')}
\bigl(U_j^X[t_j,A^1]\bigr), \quad
U_j^{X'}[t_j',A^4]
=
\ell^j_{(X,t_j)\to (X',t_j')}
\bigl(U_j^X[t_j,A^2]\bigr).
\]
Since the affine transformation is strictly increasing, it follows that $U_j^{X'}[t_j',A^3]\ge U_j^{X'}[t_j',A^4]$, so the transfer produces only true comparisons. Second, note \(\mathcal L^*\) is closed under symmetry because, by \eqref{eq:affineinverse}, 
\[
U_j^{X'}[t_j',A']
=
\ell^j_{(X,t_j)\to (X',t_j')}
\bigl(U_j^X[t_j,A]\bigr)
\quad \Leftrightarrow \quad
U_j^X[t_j,A]
=
\ell^j_{(X',t_j')\to (X,t_j)}
\bigl(U_j^{X'}[t_j',A']\bigr).
\]

Third, consider mixtures. Let \((S,\mu)\) be a probability space and suppose
that, for every \(s\in S\),
\[
\bigl((X,t_j,A^s)\indiffu{j}(X',t_j',{A'}^s)\bigr)
\in \mathcal L^*.
\]
Then, for every \(s\in S\) we have $U_j^{X'}[t_j',{A'}^s]
=
\ell^j_{(X,t_j)\to (X',t_j')}
\bigl(U_j^X[t_j,A^s]\bigr).$ Since expected utility is affine in the mixed action profile and
\(\ell^j_{(X,t_j)\to (X',t_j')}\) is affine,
\[
\begin{aligned}
U_j^{X'}\left[t_j',\int_S {A'}^s\,d\mu\right]
&=
\int_S U_j^{X'}[t_j',{A'}^s]\,d\mu  \\
&=
\int_S
\ell^j_{(X,t_j)\to (X',t_j')}
\bigl(U_j^X[t_j,A^s]\bigr)\,d\mu  \\
&=
\ell^j_{(X,t_j)\to (X',t_j')}
\left(
\int_S U_j^X[t_j,A^s]\,d\mu
\right) =
\ell^j_{(X,t_j)\to (X',t_j')}
\left(
U_j^X\left[t_j,\int_S A^s\,d\mu\right]
\right).
\end{aligned}
\]
Thus, the mixed analogy also belongs to \(\mathcal L^*\). Hence
\(\mathcal L^*\) contains \(\overline{\mathcal L}\). The transfer argument
above therefore applies to every pair of analogies in
\(\overline{\mathcal L}\), so \(\mathfrak R^{\mathrm{true}}\) is closed under
\(\mathcal L\).

Now, fix any \(\mathcal R\subseteq \mathfrak R^{\mathrm{true}}\). Then the smallest
comparison set containing \(\mathcal R\) and closed under \(\mathcal L\) is
contained in \(\mathfrak R^{\mathrm{true}}\). Hence \(\mathcal L\) is valid.

\((\Rightarrow)\). Suppose \(\mathcal L\) is valid and let
\(\overline{\mathcal L}\) be the smallest set of analogies containing \(\mathcal L\) and closed under the symmetry and mixture properties. Fix an agent \(j\), mechanisms \(X=(\mathcal A,\Phi)\) and
\(X'=(\mathcal A',\Phi')\), and types \(t_j,t_j'\in\mathcal T_j\). Define
\[
S
:=
\Big\{
\left(U_j^X[t_j,A],U_j^{X'}[t_j',A']\right)
:
\bigl((X,t_j,A)\indiffu{j}(X',t_j',A')\bigr)
\in \overline{\mathcal L}
\Big\}
\subseteq \mathbb R^2.
\]

\emph{Step 1: validity implies order preservation.}
Take any \((x_1,y_1),(x_2,y_2)\in S\) and let their corresponding analogies in \(\overline{\mathcal L}\) be $(X,t_j,A^1)\indiffu{j}(X',t_j',A^3)$ and $(X,t_j,A^2)\indiffu{j}(X',t_j',A^4),$ that is,
\[
x_1=U_j^X[t_j,A^1],\quad x_2=U_j^X[t_j,A^2],\quad
y_1=U_j^{X'}[t_j',A^3],\quad y_2=U_j^{X'}[t_j',A^4].
\]
We claim that $x_1\ge x_2$ if and only if $y_1\ge y_2.$ Suppose first that \(x_1\ge x_2\). Then $U_j^X[t_j,A^1]\ge U_j^X[t_j,A^2]$ is a true comparison. Let \(\mathcal R\) be the singleton set containing this comparison. Now, let \(\hat{\mathcal R}\) be the smallest comparison set containing
\(\mathcal R\) and closed under \(\mathcal L\). Since the two analogies
above belong to \(\overline{\mathcal L}\), the closure property in
Definition~\ref{def:L-closed} implies that
\[
\bigl(U_j^{X'}[t_j',A^3]\ge U_j^{X'}[t_j',A^4]\bigr)
\in \hat{\mathcal R}.
\]
By validity of \(\mathcal L\), \(\hat{\mathcal R}\subseteq
\mathfrak R^{\mathrm{true}}\). Hence this comparison is true, so
\(y_1\ge y_2\).

Conversely, suppose \(y_1\ge y_2\). Then $U_j^{X'}[t_j',A^3]\ge U_j^{X'}[t_j',A^4]$ is a true comparison. An analogous argument then holds because the reverse analogies, $(X',t_j',A^3)\indiffu{j}(X,t_j,A^1)$ and $(X',t_j',A^4)\indiffu{j}(X,t_j,A^2)$, belong to \(\overline{\mathcal L}\) by the symmetry property. 

\emph{Step 2: order preservation makes \(S\) the graph of a function.} Observe that if two points in \(S\) have the same first coordinate, then they have the same
second coordinate: if \(x_1=x_2\), then both \(x_1\ge x_2\) and
\(x_2\ge x_1\), so Step 1 gives \(y_1\ge y_2\) and \(y_2\ge y_1\). Thus
\(y_1=y_2\). Hence there is a set \(D\subseteq\mathbb R\) and a function
\(f:D\to\mathbb R\) such that $S=\{(x,f(x)):x\in D\}.$ Moreover, Step 1 implies that \(f\) is strictly increasing whenever \(D\)
contains more than one point.

\emph{Step 3: mixtures force the induced function to be affine.}
If \(D\) is empty, choose $\ell^j_{(X,t_j)\to (X',t_j')}(z):=z.$ If \(D\) is a singleton, say \(D=\{x_0\}\) and \(f(x_0)=y_0\), choose any
\(\kappa>0\) and set
\[
\ell^j_{(X,t_j)\to (X',t_j')}(z):=\kappa z+(y_0-\kappa x_0)
\quad
\implies
\quad
\ell^j_{(X,t_j)\to (X',t_j')}(x_0)=y_0.
\]

Now suppose \(D\) contains at least two points. Take distinct
\(x_0,x_1\in D\), and choose analogies in \(\overline{\mathcal L}\)
corresponding to $(x_0,f(x_0))$ and $(x_1,f(x_1))$. For any \(\alpha\in[0,1]\), the mixture property gives $\left(
\alpha x_0+(1-\alpha)x_1,\,
\alpha f(x_0)+(1-\alpha)f(x_1)
\right)
\in S.$ Therefore $f\bigl(\alpha x_0+(1-\alpha)x_1\bigr)
=
\alpha f(x_0)+(1-\alpha)f(x_1).$ Assume without loss of generality that \(x_0<x_1\), and define
\[
\kappa:=
\frac{f(x_1)-f(x_0)}{x_1-x_0},
\qquad
\lambda:=f(x_0)-\kappa x_0.
\]
Since \(f\) is strictly increasing, \(\kappa>0\). We claim that $f(x)=\kappa x+\lambda$ for every $x\in D.$

If \(x\in[x_0,x_1]\cap D\), then \(x=\alpha x_1+(1-\alpha)x_0\) for some
\(\alpha\in[0,1]\), so
\[
f(x)=\alpha f(x_1)+(1-\alpha)f(x_0)=\kappa x+\lambda.
\]
If \(x>x_1\), then \(x_1=\alpha x+(1-\alpha)x_0\) for some
\(\alpha\in(0,1)\). Hence $f(x_1)=\alpha f(x)+(1-\alpha)f(x_0),$ which implies \(f(x)=\kappa x+\lambda\). The case \(x<x_0\) is identical,
using that \(x_0\) is a convex combination of \(x\) and \(x_1\). Thus \(f\)
agrees on all of \(D\) with the positive affine function
\[
\ell^j_{(X,t_j)\to (X',t_j')}(z):=\kappa z+\lambda.
\]

\emph{Step 4: the affine map represents every original analogy.}
Since \(\mathcal L\subseteq\overline{\mathcal L}\), every original analogy $(X,t_j,A)\indiffu{j}(X',t_j',A')$ in \(\mathcal L\) satisfies $U_j^{X'}[t_j',A']
=
\ell^j_{(X,t_j)\to (X',t_j')}
\bigl(U_j^X[t_j,A]\bigr).$

\emph{Step 5: the reverse affine map can be chosen as the inverse.}
Finally, choose the affine map for the reverse ordered pair to be the inverse:
\[
\ell^j_{(X',t_j')\to (X,t_j)}
=
\left(
\ell^j_{(X,t_j)\to (X',t_j')}
\right)^{-1}.
\]
This is valid because symmetry puts the reverse analogies in \(\overline{\mathcal L}\), and the
order-preservation argument above identifies exactly the inverse payoff
relation. In the case \((X,t_j)=(X',t_j')\), the same argument applied
to an analogy and its symmetric counterpart implies that every
self-analogy in \(\overline{\mathcal L}\) preserves payoffs exactly; we may
therefore choose the identity map, which is its own inverse.

\subsection{Proof of Theorem \ref{th:strategic-equivalence-transfer}}

Let \(\sigma\) be an equilibrium of \(X\) for prior \(F\). First note that, since \(X\) and \(X'\) are strategically equivalent through
\(\alpha\), for every agent \(j\), type \(t_j\), and mixed action profile
\(A\in\Delta(\mathcal A)\), we have $U_j^{X'}[t_j,\alpha_\# A]=U_j^X[t_j,A].$ Hence \(\sigma'\) is an equilibrium of \(X'\) for prior \(F\):
every deviation in \(X'\) is the \(\alpha_j\)-pushforward of a deviation in
\(X\), and the relevant expected payoffs coincide.

This proves part~$(i)$. For part~$(ii)$, impose the additional common-knowledge assumptions in the theorem. Let
\[
\mathcal L^\alpha
:=
\big\{(X,t_i,a)\indiffu{i}(X',t_i,\alpha(a)):
i\in\mathcal I,\ t_i\in\mathcal T_i,\ a\in\mathcal A\big\}.
\]
Fix agents \(i,h,j\), an entertainment set \(\mathcal E\in\mathfrak E_i\), an awareness profile \(\omega\in\mathcal E\), a type \(t_j\in\mathcal T_j\), and a mixed deviation \(A_j'\in\Delta(\mathcal A_j')\). Since \(\alpha_j\) is a bijection, there exists
\(A_j\in\Delta(\mathcal A_j)\) such that $A_j'=(\alpha_j)_\# A_j.$ By construction of \(\sigma'\), we have $A_{-j}(\sigma',F\mid t_j)
=
(\alpha_{-j})_\# A_{-j}(\sigma,F\mid t_j).$ Since it is commonly known that \(\sigma\) is an equilibrium of \(X\) for prior
\(F\), we have
\[
\Bigl(
U_j^X\bigl[t_j,\sigma_j(t_j)\otimes A_{-j}(\sigma,F\mid t_j)\bigr]
\ge
U_j^X\bigl[t_j,A_j\otimes A_{-j}(\sigma,F\mid t_j)\bigr]
\Bigr)
\in \mathcal{R}_h^\omega .
\]
It is common knowledge that agents reason according to
\(\mathcal L^\alpha\), so \(\mathcal R_h^\omega\) is closed under
\(\mathcal L^\alpha\). The pure-profile analogies in
\(\mathcal L^\alpha\), together with the mixture property defining
\(\overline{\mathcal L^\alpha}\), imply that
\[
\bigl(
X,t_j,\sigma_j(t_j)\otimes A_{-j}(\sigma,F\mid t_j)
\bigr)
\indiffu{j}
\bigl(
X',t_j,\sigma_j'(t_j)\otimes A_{-j}(\sigma',F\mid t_j)
\bigr)
\]
and
\[
\bigl(
X,t_j,A_j\otimes A_{-j}(\sigma,F\mid t_j)
\bigr)
\indiffu{j}
\bigl(
X',t_j,A_j'\otimes A_{-j}(\sigma',F\mid t_j)
\bigr)
\]
belong to \(\overline{\mathcal L^\alpha}\). Since
\(\mathcal R_h^\omega\) is closed under
\(\mathcal L^\alpha\), applying
Definition~\ref{def:L-closed} gives
\[
\Bigl(
U_j^{X'}\bigl[t_j,\sigma_j'(t_j)\otimes
A_{-j}(\sigma',F\mid t_j)\bigr]
\ge
U_j^{X'}\bigl[t_j,A_j'\otimes A_{-j}(\sigma',F\mid t_j)\bigr]
\Bigr)
\in \mathcal{R}_h^\omega .
\]

\subsection{Proof of Theorem \ref{th:strategic-analogy-transfer}}

Fix \(F\in\Delta(\mathcal T)\), and let \(\sigma\) be an equilibrium of \(X\) for prior \(\tau_\#F\).

First, \(\sigma'\) is an equilibrium of \(X'\) for prior \(F\). Indeed, by
strategic analogy and Proposition~\ref{th1:eqchar}, for each agent \(j\) and
type \(t_j\), the payoff in \(X'\) from the \(\alpha\)-image of any mixed profile is a positive affine transformation of the corresponding payoff in \(X\) for type \(\tau_j(t_j)\). Thus, the set of best responses is preserved under
\(\alpha_j\). Since \(\sigma\) is an equilibrium of \(X\) for prior
\(\tau_\#F\), the profile \(\sigma'\) is an equilibrium of \(X'\) for prior
\(F\).

This proves part~$(i)$. For part~$(ii)$, impose the additional common-knowledge assumptions in the theorem. Fix agents \(i,h,j\), an entertainment set \(\mathcal E\in\mathfrak E_i\), an awareness profile \(\omega\in\mathcal E\), a type \(t_j\in\mathcal T_j\), and a mixed deviation \(A_j'\in\Delta(\mathcal A_j')\).
Since \(\alpha_j\) is a bijection, there exists
\(A_j\in\Delta(\mathcal A_j)\) such that $A_j'=(\alpha_j)_\# A_j.$ Moreover, by construction of \(\sigma'\) and by the definition of the
pushforward prior \(\tau_\#F\),
\[
A_{-j}(\sigma',F\mid t_j)
=
(\alpha_{-j})_\#
A_{-j}(\sigma,\tau_\#F\mid \tau_j(t_j)).
\]

Since it is commonly known that \(\sigma\) is an equilibrium of \(X\) for prior
\(\tau_\#F\), we have
\[
\Bigl(
U_j^X\bigl[
\tau_j(t_j),
\sigma_j(\tau_j(t_j))
\otimes
A_{-j}(\sigma,\tau_\#F\mid \tau_j(t_j))
\bigr]
\ge
U_j^X\bigl[
\tau_j(t_j),
A_j
\otimes
A_{-j}(\sigma,\tau_\#F\mid \tau_j(t_j))
\bigr]
\Bigr)
\in \mathcal{R}_h^\omega .
\]
Since it is common knowledge that agents reason according to
\(\mathcal L_{\alpha,\tau}^{X,X'}\), the set \(\mathcal R_h^\omega\) is
closed under \(\mathcal L_{\alpha,\tau}^{X,X'}\). The pure-profile
analogies in \(\mathcal L_{\alpha,\tau}^{X,X'}\), together with the
mixture property defining \(\overline{\mathcal L_{\alpha,\tau}^{X,X'}}\),
imply that
\[
\bigl(
X,
\tau_j(t_j),
\sigma_j(\tau_j(t_j))
\otimes
A_{-j}(\sigma,\tau_\#F\mid \tau_j(t_j))
\bigr)
\indiffu{j}
\bigl(
X',
t_j,
\sigma_j'(t_j)
\otimes
A_{-j}(\sigma',F\mid t_j)
\bigr)
\]
and
\[
\bigl(
X,
\tau_j(t_j),
A_j
\otimes
A_{-j}(\sigma,\tau_\#F\mid \tau_j(t_j))
\bigr)
\indiffu{j}
\bigl(
X',
t_j,
A_j'
\otimes
A_{-j}(\sigma',F\mid t_j)
\bigr)
\]
belong to \(\overline{\mathcal L_{\alpha,\tau}^{X,X'}}\). Since
\(\mathcal R_h^\omega\) is closed under
\(\mathcal L_{\alpha,\tau}^{X,X'}\), applying
Definition~\ref{def:L-closed} gives
\[
\Bigl(
U_j^{X'}\bigl[
t_j,
\sigma_j'(t_j)\otimes A_{-j}(\sigma',F\mid t_j)
\bigr]
\ge
U_j^{X'}\bigl[
t_j,
A_j'\otimes A_{-j}(\sigma',F\mid t_j)
\bigr]
\Bigr)
\in \mathcal{R}_h^\omega .
\]

\subsection{Proof of Proposition \ref{fact:dominant_exists_preserved}}

Let \(X=(\mathcal A,\Phi)\) and \(X'=(\mathcal A',\Phi')\), and suppose that
\(X\) admits dominant strategies for all agents. Let $\alpha, \tau$ and $\kappa, \lambda$ be the bijections and functions from the equivalent definition of strategic analogy. Then, for every agent \(i\), every type \(t_i\in\mathcal T_i\), and every
pure action profile \(a\in\mathcal A\),
\[
U_i^{X'}[t_i,\alpha(a)]
=
\kappa_i(t_i)U_i^X[\tau_i(t_i),a]+\lambda_i(t_i).
\]
Since \(X\) admits dominant strategies for all agents, for each agent \(i\)
choose a dominant strategy \(s_i:\mathcal T_i\to\Delta(\mathcal A_i)\), and
define \(s_i':\mathcal T_i\to\Delta(\mathcal A_i')\) by $s_i'(t_i):=
(\alpha_i)_{\#}s_i\bigl(\tau_i(t_i)\bigr).$ We show that \(s_i'\) is dominant for agent \(i\) in mechanism \(X'\). Fix
\(t_i\in\mathcal T_i\) and
\(\mu_{-i}'\in\Delta(\mathcal A_{-i}')\), and let $\mu_{-i}
:=
(\alpha_{-i}^{-1})_{\#}\mu_{-i}'.$ By strategic analogy and linearity, for every
\(A_i\in\Delta(\mathcal A_i)\),
\[
U_i^{X'}
\bigl[
t_i,
(\alpha_i)_{\#}A_i\otimes \mu_{-i}'
\bigr]
=
\kappa_i(t_i)
U_i^X
\bigl[
\tau_i(t_i),
A_i\otimes \mu_{-i}
\bigr]
+
\lambda_i(t_i).
\]
Because \(s_i\) is dominant for agent \(i\) in mechanism \(X\), for every
\(A_i\in\Delta(\mathcal A_i)\),
\[
U_i^X
\bigl[
\tau_i(t_i),
s_i(\tau_i(t_i))\otimes \mu_{-i}
\bigr]
\ge
U_i^X
\bigl[
\tau_i(t_i),
A_i\otimes \mu_{-i}
\bigr].
\]
Since \(\kappa_i(t_i)>0\), multiplying by \(\kappa_i(t_i)\) and adding
\(\lambda_i(t_i)\) preserves the inequality. Translating back to mechanism
\(X'\), this gives
\[
U_i^{X'}
\bigl[
t_i,
s_i'(t_i)\otimes \mu_{-i}'
\bigr]
\ge
U_i^{X'}
\bigl[
t_i,
(\alpha_i)_{\#}A_i\otimes \mu_{-i}'
\bigr]
\]
for every \(A_i\in\Delta(\mathcal A_i)\). Fix any
\(A_i'\in\Delta(\mathcal A_i')\) and let $A_i:=
(\alpha_i^{-1})_{\#}A_i'.$ Then \((\alpha_i)_{\#}A_i=A_i'\), so
\[
U_i^{X'}
\bigl[
t_i,
s_i'(t_i)\otimes \mu_{-i}'
\bigr]
\ge
U_i^{X'}
\bigl[
t_i,
A_i'\otimes \mu_{-i}'
\bigr].
\]
Since \(t_i\), \(\mu_{-i}'\), and \(A_i'\) were arbitrary, \(s_i'\) is dominant for agent \(i\) in mechanism \(X'\). The converse follows by interchanging \(X\) and \(X'\).

\subsection{Proof of Proposition~\ref{prop:fpamaxclass}}

\((\Rightarrow)\). Let \(X=(\mathcal A,\Phi)\) be a standard auction and suppose that \(X\) is
strategically analogous to \(k\mathrm{PA}(r)\). By the affine formulation of
strategic analogy, there exist bijections $\alpha_i:\{\varnothing\}\cup[r,\infty)
\to
\{\varnothing\}\cup[\underline a,\infty)$, $\tau_i:\mathbb R\to\mathbb R,$ and functions $\kappa_i:\mathbb R\to\mathbb R_{++},$ $\lambda_i:\mathbb R\to\mathbb R$
such that, for every bidder \(i\), every type \(t\), and every profile
\(a\in(\{\varnothing\}\cup[r,\infty))^N\),
\begin{equation}\label{eq:affine-fpa-max}
U_i^X[t,\alpha(a)]
=
\kappa_i(t)U_i^{k\mathrm{PA}(r)}[\tau_i(t),a]
+\lambda_i(t).
\end{equation}

\emph{Step 1: \(\alpha_i(\varnothing)=\varnothing\) and
\(\lambda_i\equiv0\).}
Fix bidder \(i\). If \(a_i=\varnothing\), then $U_i^{k\mathrm{PA}(r)}[\tau_i(t),a]=0$, so $U_i^X[t,\alpha_i(\varnothing),\alpha_{-i}(a_{-i})]
=
\lambda_i(t)$ for all $t,a_{-i}.$ Since each \(\alpha_j\) is onto,
\begin{equation}\label{eq:lambda-any-fpa-max}
U_i^X[t,\alpha_i(\varnothing),b_{-i}]
=
\lambda_i(t)
\quad\text{for all }t,b_{-i}.
\end{equation}
Suppose \(\alpha_i(\varnothing)\neq\varnothing\). If bidder \(i\) chooses
\(\alpha_i(\varnothing)\) and all other bidders choose \(\varnothing\), then
bidder \(i\) is the unique highest participating bidder, so her allocation
probability is one. Thus the left side of \eqref{eq:lambda-any-fpa-max} has
slope one in \(t\). If instead every bidder chooses
\(\alpha_i(\varnothing)\), then all bidders are tied for the highest bid, so
bidder \(i\)'s allocation probability is \(1/N\). The left side then has slope
\(1/N\) in \(t\). Since \(N\geq2\), this contradicts
\eqref{eq:lambda-any-fpa-max}. Hence $\alpha_i(\varnothing)=\varnothing.$ Substituting back into \eqref{eq:lambda-any-fpa-max} and using the outside
option in \(X\), we get
\begin{equation}\label{eq:lambzeroo}
\lambda_i(t)=0
\quad\text{for all }t.
\end{equation}

\emph{Step 2: zero allocation in \(k\mathrm{PA}(r)\) implies zero allocation
and zero payment in \(X\).}
Write \(q_i^Y(a)\) and \(p_i^Y(a)\) for bidder \(i\)'s allocation probability
and expected payment under mechanism \(Y\) at profile \(a\). If $q_i^{k\mathrm{PA}(r)}(a)=0,$ then $U_i^{k\mathrm{PA}(r)}[\tau_i(t),a]=0$ for all $t$. Using \eqref{eq:affine-fpa-max} and \eqref{eq:lambzeroo}, this gives $U_i^X[t,\alpha(a)]=0$ for all $t$. Since $U_i^X[t,\alpha(a)]
= tq_i^X(\alpha(a))-p_i^X(\alpha(a)),$ we obtain
\begin{equation}\label{eq:zero-fpa-max}
q_i^{k\mathrm{PA}(r)}(a)=0
\quad\Longrightarrow\quad
q_i^X(\alpha(a))=p_i^X(\alpha(a))=0.
\end{equation}

\emph{Step 3: \(\kappa_i\) is constant and \(\tau_i\) is affine.}
Fix bidder \(i\). Choose two profiles \(a,\tilde a\) at which bidder \(i\)
strictly wins in \(k\mathrm{PA}(r)\), with different payments $\pi^r(a)\neq \pi^r(\tilde a).$ At these profiles, \eqref{eq:affine-fpa-max} and \eqref{eq:lambzeroo} give:
\[
tq_i^X(\alpha(a))-p_i^X(\alpha(a))
=
\kappa_i(t)\bigl(\tau_i(t)-\pi^r(a)\bigr),
\]
\[
tq_i^X(\alpha(\tilde a))-p_i^X(\alpha(\tilde a))
=
\kappa_i(t)\bigl(\tau_i(t)-\pi^r(\tilde a)\bigr).
\]
Subtracting gives $t\Delta q_i-\Delta p_i
=
\kappa_i(t)\bigl(\pi^r(\tilde a)-\pi^r(a)\bigr),$ where \(\Delta q_i\) and \(\Delta p_i\) are independent of \(t\). Hence
\(\kappa_i\) is affine in \(t\). Since \(\kappa_i(t)>0\) for every
\(t\in\mathbb R\), it must be constant: $\kappa_i(t)\equiv c_i>0.$ Thus, whenever bidder \(i\) strictly wins in \(k\mathrm{PA}(r)\) at profile
\(a\),
\begin{equation}\label{eq:strict-win-fpa-max}
tq_i^X(\alpha(a))-p_i^X(\alpha(a))
=
c_i\tau_i(t)-c_i\pi^r(a).
\end{equation}
The left side is affine in \(t\), so \(\tau_i\) is affine. Since \(\tau_i\) is a
bijection of \(\mathbb R\), write $\tau_i(t)=m_i t+n_i$ with $m_i\neq0.$

\emph{Step 4: strict winners in \(k\mathrm{PA}(r)\) receive the object with
probability one in \(X\).}
If bidder \(i\) strictly wins in \(k\mathrm{PA}(r)\) at profile \(a\), then
every other bidder has allocation probability zero in \(k\mathrm{PA}(r)\). By
\eqref{eq:zero-fpa-max}, every other bidder also has allocation probability
zero in \(X\) at \(\alpha(a)\). Since
\(\alpha_i(\varnothing)=\varnothing\) and \(\alpha_i\) is injective, bidder
\(i\)'s image action is not \(\varnothing\). Hence bidder \(i\) is the only
bidder who can win in \(X\), and the object is allocated with
probability one. Comparing coefficients in \eqref{eq:strict-win-fpa-max}, we
get
\begin{equation}\label{eq:payment-strict-fpa-max}
c_i m_i=1,
\qquad
p_i^X(\alpha(a))=c_i\bigl(\pi^r(a)-n_i\bigr)
\end{equation}
whenever bidder \(i\) strictly wins in \(k\mathrm{PA}(r)\).

\emph{Step 5: the action relabelings are common across bidders.}
Suppose not. Then for some bidders \(i,j\) and some non-null bid
\(x\in[\underline a,\infty)\), $\alpha_i^{-1}(x)\neq \alpha_j^{-1}(x).$ Consider the \(X\)-profile in which bidders \(i\) and \(j\) both choose \(x\)
and all other bidders choose \(\varnothing\). In the corresponding
\(k\mathrm{PA}(r)\)-profile, bidders \(i\) and \(j\) choose two distinct
non-null bids. Hence exactly one of them strictly wins and the other loses. By
\eqref{eq:zero-fpa-max} and \eqref{eq:payment-strict-fpa-max}, the image
profile in \(X\) gives allocation probability one to the strict winner and zero
to the other bidder. But in \(X\), bidders \(i\) and \(j\) are tied for the
highest bid \(x\), so each receives allocation probability \(1/2\); contradiction. Therefore $\alpha_i=\alpha_j$ for all $i,j$. Write the common relabeling as \(\alpha_0\).

\emph{Step 6: the constants \(c_i\) and \(n_i\) are common across bidders.}
Take bidders \(i\) and \(j\). Consider two strict-win profiles that are
permutations of each other, one in which \(i\) wins and one in which \(j\) wins,
with the same payment \(\pi\). Since \(\alpha_0\) is common and
payments in \(X\) are invariant under permutations of bidder labels, the
winner's expected payment in \(X\) is the same in the two image profiles. By
\eqref{eq:payment-strict-fpa-max}, $c_i(\pi-n_i)=c_j(\pi-n_j).$ Since this equality holds for at least two distinct values of \(\pi\), we get $c_i=c_j$ and $n_i=n_j.$ Thus, there are constants \(c>0\) and \(n\in\mathbb R\) such that $c_i=c,$ $n_i=n,$ and, by \eqref{eq:payment-strict-fpa-max}, $m_i=1/c$ for every bidder \(i\).

\emph{Step 7: construction of the equivalent \(k\)th-price auction.}
Fix any profile \(a\) in \(k\mathrm{PA}(r)\). If
\(q_i^{k\mathrm{PA}(r)}(a)=0\), then \eqref{eq:zero-fpa-max} gives $q_i^X(\alpha(a))=p_i^X(\alpha(a))=0.$ If \(q_i^{k\mathrm{PA}(r)}(a)>0\), then
\[
U_i^{k\mathrm{PA}(r)}[\tau_i(t),a]
=
q_i^{k\mathrm{PA}(r)}(a)\bigl(\tau_i(t)-\pi^r(a)\bigr).
\]
Using \(\kappa_i\equiv c\), \(\tau_i(t)=t/c+n\), and \(\lambda_i\equiv0\),
equation \eqref{eq:affine-fpa-max} gives
\[
tq_i^X(\alpha(a))-p_i^X(\alpha(a))
=
c q_i^{k\mathrm{PA}(r)}(a)
\big(\tfrac{t}{c}+n-\pi^r(a)\big).
\]
Therefore,
\begin{equation}\label{eq:general-payment-fpa-max}
q_i^X(\alpha(a))
=
q_i^{k\mathrm{PA}(r)}(a),
\qquad
p_i^X(\alpha(a))
=
c q_i^{k\mathrm{PA}(r)}(a)\bigl(\pi^r(a)-n\bigr).
\end{equation}

Because payments in \(X\) are nonnegative, evaluating
\eqref{eq:general-payment-fpa-max} at a profile whose payment-relevant price is
\(r\) gives $c(r-n)\geq0.$ Define $r':=c(r-n)\in\mathbb R_+.$ Now define a bijection $\beta_i:\{\varnothing\}\cup[r',\infty)
\to
\{\varnothing\}\cup[\underline a,\infty)$ by $\beta_i(\varnothing)=\varnothing$ and $\beta_i(b)=\alpha_0\big(\tfrac{b}{c}+n\big)$ for $b\geq r'.$ This is well-defined because \(b\geq r'=c(r-n)\) implies $\frac{b}{c}+n\geq r.$ Take any profile \(b\) in \(k\mathrm{PA}(r')\), and define a profile \(a\) in
\(k\mathrm{PA}(r)\) by
\[
a_i=
\begin{cases}
\varnothing, & b_i=\varnothing,\\[3pt]
b_i/c+n, & b_i\neq\varnothing.
\end{cases}
\]
The map \(b_i\mapsto b_i/c+n\) is strictly increasing, so it preserves
participation, rankings, highest-bidder sets, and allocation probabilities.
Moreover, if \(\pi^{r'}(b)\) is the payment-relevant price in
\(k\mathrm{PA}(r')\), then $\pi^{r'}(b)=c\bigl(\pi^r(a)-n\bigr).$ Using \eqref{eq:general-payment-fpa-max}, for every bidder \(i\),
\[
q_i^X(\beta(b))
=
q_i^{k\mathrm{PA}(r')}(b),
\qquad
p_i^X(\beta(b))
=
p_i^{k\mathrm{PA}(r')}(b).
\]
Hence, for every bidder \(i\), every type \(t_i\), and every profile \(b\), we have $U_i^X[t_i,\beta(b)]
=
U_i^{k\mathrm{PA}(r')}[t_i,b].$ Therefore \(X\) is strategically equivalent to \(k\mathrm{PA}(r')\).

\((\Leftarrow)\) Suppose \(X\) is strategically equivalent to
\(k\mathrm{PA}(r')\) for some \(r'\in\mathbb R_+\). Then there are bijections $\beta_i:\{\varnothing\}\cup[r',\infty)
\to
\{\varnothing\}\cup[\underline a,\infty)$ such that, for every bidder \(i\), every type \(t_i\), and every profile
\(b\in(\{\varnothing\}\cup[r',\infty))^N\),
\begin{equation}\label{eq:equiv-converse}
U_i^X[t_i,\beta(b)]
=
U_i^{k\mathrm{PA}(r')}[t_i,b].
\end{equation}

We show that \(X\) is strategically analogous to \(k\mathrm{PA}(r)\). Define
\[
\gamma_i(\varnothing)=\varnothing,
\qquad
\gamma_i(a_i)=a_i+(r'-r)
\quad\text{for }a_i\in[r,\infty).
\]
Then \(\gamma_i\) is a bijection from
\(\{\varnothing\}\cup[r,\infty)\) to
\(\{\varnothing\}\cup[r',\infty)\). The map \(\gamma\) preserves
participation, rankings of bids, highest-bidder sets, and allocation
probabilities. Moreover, whenever the object is allocated, the payment-relevant
\(k\)th price in \(k\mathrm{PA}(r')\) at \(\gamma(a)\) is the payment-relevant
\(k\)th price in \(k\mathrm{PA}(r)\) at \(a\) plus \(r'-r\). Therefore, for
every bidder \(i\), type \(t_i\), and profile \(a\),
\begin{equation}\label{eq:kpa-analogy-converse}
U_i^{k\mathrm{PA}(r')}[t_i,\gamma(a)]
=
U_i^{k\mathrm{PA}(r)}[t_i-(r'-r),a].
\end{equation}

Now, define $\alpha_i:=\beta_i\circ\gamma_i,$ $\tau_i(t)=t-(r'-r),$ $\kappa_i(t)=1,$ and $\lambda_i(t)=0.$ Combining \eqref{eq:equiv-converse} and \eqref{eq:kpa-analogy-converse}, we
get $U_i^X[t_i,\alpha(a)]
=
U_i^{k\mathrm{PA}(r)}[\tau_i(t_i),a]$ for every bidder \(i\), type \(t_i\), and profile \(a\).
Thus \(X\) is strategically analogous to \(k\mathrm{PA}(r)\).

\subsection{Proof of Proposition \ref{prop:entry-costs}}

The negative result follows directly from Proposition~\ref{prop:fpamaxclass}. To show the positive result, fix \(c,c'>0\) and define, for each bidder \(i\),
\[
\alpha_i(\varnothing)=\varnothing,
\quad
\alpha_i(b)=\frac{c'}{c}\, b
\;\text{for } b\in\mathbb{R}_+,
\quad\text{and}\quad
\tau_i(t)=\frac{c}{c'}t,
\quad
\kappa_i(t)=\frac{c'}{c},
\quad
\lambda_i(t)=0.
\]
Each \(\alpha_i\) is a bijection from \(\{\varnothing\}\cup\mathbb{R}_+\) to
itself, and each \(\tau_i\) is a bijection from \(\mathbb R\) to itself.
Fix a bidder \(i\), a type \(t_i\in\mathbb R\), and an action profile \(a\).
The map \(\alpha\) preserves the set of participants and the ranking of bids.
Hence it preserves the allocation probabilities induced by the symmetric
tie-breaking rule. Moreover, if \(a_i\neq\varnothing\), then the \(k\)th-highest
participating bid is multiplied by \(c'/c\):
\[
\alpha(a)^{(k)}=\frac{c'}{c}\, a^{(k)}.
\]
If \(a_i=\varnothing\), then bidder \(i\)'s payoff is zero in both mechanisms,
so
\[
U_i^{k\mathrm{EC}(c')}[t_i,\alpha(a)]
=
0
=
\frac{c'}{c}\,U_i^{k\mathrm{EC}(c)}[\tau_i(t_i),a].
\]
Now suppose \(a_i\neq\varnothing\). Let \(q_i(a)\) be bidder \(i\)'s allocation
probability in the profile \(a\). Since \(\alpha\) preserves allocation
probabilities, $q_i(\alpha(a))=q_i(a).$ Therefore
\[
U_i^{k\mathrm{EC}(c')}[t_i,\alpha(a)]
=
q_i(a)\Bigl(t_i-\frac{c'}{c}\, a^{(k)}\Bigr)-c' =
\frac{c'}{c}\left[
q_i(a)\left(\frac{c\,t_i}{c'}-a^{(k)}\right)-c
\right] =
\frac{c'}{c}\,U_i^{k\mathrm{EC}(c)}[\tau_i(t_i),a].
\]
Thus, for every bidder \(i\), type \(t_i\), and action profile \(a\),
\[
U_i^{k\mathrm{EC}(c')}[t_i,\alpha(a)]
=
\kappa_i(t_i)\,U_i^{k\mathrm{EC}(c)}[\tau_i(t_i),a]
+\lambda_i(t_i),
\]
with \(\kappa_i(t_i)=c'/c>0\) and \(\lambda_i(t_i)=0\). Hence
\(k\mathrm{EC}(c)\) and \(k\mathrm{EC}(c')\) are strategically analogous.

\subsection{Proof of Proposition \ref{prop:no-class-all-myerson}}

We begin with the following lemma.

\begin{lemma}\label{lemma:distconstr}
There exists a distribution \(F\) with the following properties:
\begin{enumerate}
    \item \(F\) is absolutely continuous, has support \([0,1]\), and has a positive density a.e. on \((0,1)\).

    \item There are ironing intervals \(I_m=[\ell_m,h_m]\), indexed by \(m\ge2\), with \(0<\ell_m<h_m<\ell_{m+1}<1\) and \(\ell_m\to1\). Ironing is strict throughout the interior of each \(I_m\).

    \item The ironed virtual value takes a constant value \(c_m>0\) on the interior of each \(I_m\), with \(c_m<c_{m+1}\), and is strictly increasing on \((0,1)\setminus\bigcup_{m\ge2}I_m\). It is strictly negative on \((0,\varepsilon)\) for some \(\varepsilon>0\).

    \item Writing \(\rho_m:=F(h_m)-F(\ell_m)\), we have \(F(\ell_m)\to1\) and \((h_m-\ell_m)/\rho_m\to\infty\).
\end{enumerate} 
\end{lemma}
\begin{proof}
We construct such a prior by specifying its revenue curve. Let \(u=1-F(v)\) denote the upper-tail quantile and \(v(u)=F^{-1}(1-u)\) be the corresponding value.  The revenue curve is then given by \(R(u)=uv(u)\); conversely, a revenue curve determines the distribution through \(v(u)=R(u)/u\).

First, consider the strictly concave curve \(R_0(u)=u(1-\sqrt u)\) for \(u\in[0,1]\). For each \(m\ge2\), define \(u_m:=4^{-m}\) and \(J_m:=[u_m,2u_m]\); note that these intervals are disjoint. Now, for each \(m\ge2\), let \(L_m\) be the unique affine function satisfying \(L_m(u_m)=R_0(u_m)\) and \(L_m(2u_m)=R_0(2u_m)\). Define
\[
R(u)=
\begin{cases}
L_m(u)-\dfrac1{10}u_m^{3/2}x(1-x),
&u\in J_m,\quad x=(u-u_m)/u_m,\\[4pt]
R_0(u),&u\notin\displaystyle\bigcup_{m\ge2}J_m.
\end{cases}
\]
Thus, \(R=L_m\) at the endpoints of \(J_m\) and \(R<L_m\) throughout its interior (Figure \ref{fig:strict-ironing}).

\begin{figure}[ht]
\centering
\begin{tikzpicture}[x=1.25cm,y=0.7cm,scale=0.9,transform shape,>=stealth]
    \draw[->] (0,0) -- (6.5,0) node[right] {$u$};
    \draw[->] (0,0) -- (0,6.4) node[above] {$R(u)$};

    \draw[densely dotted,gray] (2,0) -- (2,4.4);
    \draw[densely dotted,gray] (4,0) -- (4,5.6);
    \node[below] at (2,0) {$u_m$};
    \node[below] at (4,0) {$2u_m$};

    \draw[dashed,gray,thick]
        plot[domain=2:4,samples=60]
        (\x,{3*\x-0.4*\x*\x});
    \node[gray,above] at (3,5.4) {$R_0$};

    \draw[thick]
        plot[domain=0:2,samples=60]
        (\x,{3*\x-0.4*\x*\x});
    \draw[thick]
        plot[domain=4:6,samples=60]
        (\x,{3*\x-0.4*\x*\x});

    \draw[thick]
        plot[domain=2:4,samples=60]
        (\x,{0.6*\x+3.2-0.65*(\x-2)*(4-\x)});
    \node[below] at (3,4.35) {$R$};

    \draw[blue,very thick] (2,4.4) -- (4,5.6);
    \node[blue,above left] at (2.55,4.79) {$L_m$};

    \fill (2,4.4) circle (1.6pt);
    \fill (4,5.6) circle (1.6pt);

    \draw[<->] (2,-0.75) -- (4,-0.75)
        node[midway,fill=white,inner sep=2pt] {$J_m$};
\end{tikzpicture}
\caption{Schematic construction on \(J_m\). The modified revenue
curve \(R\) agrees with \(R_0\) at the endpoints and lies strictly
below \(L_m\) in the interior. Its concave envelope equals \(L_m\)
on \(J_m\), yielding a constant ironed virtual value there.}
\label{fig:strict-ironing}
\end{figure}

We now verify that \(v(u):=R(u)/u\), extended by \(v(0)=1\), defines a distribution. It is continuous, with \(v(1)=0\). To check strict decrease inside \(J_m\), write \(L_m(u)=A_m u+B_m\), where \(A_m=1-(2\sqrt2-1)\sqrt{u_m}\) and \(B_m=(2\sqrt2-2)u_m^{3/2}\). For \(0<x<1\),
\[
uR'(u)-R(u)
=u_m^{3/2}\left[-(2\sqrt2-2)
+\frac1{10}(x^2+2x-1)\right]<0.
\]
Hence \(v'(u)<0\) there. Outside these intervals, strict decrease follows from \(v(u)=1-\sqrt u\). Thus \(v\) is continuous and strictly decreasing from \(1\) to \(0\), and determines a distribution \(F\) with support \([0,1]\). We now verify properties~$(i)$--$(iv)$.

\emph{Property $(i)$.} On each smooth piece in \((0,1)\), \(v'\) is finite and strictly negative. These pieces are locally finite away from \(u=0\), so the inverse relation \(F(v(u))=1-u\) implies that \(F\) is locally absolutely continuous on \((0,1)\), with a positive density almost everywhere. Since \(F\) has no endpoint atoms, it is absolutely continuous on \([0,1]\).

\emph{Property $(ii)$.} The concave envelope of \(R\) equals \(R_0\) off the \(J_m\) and \(L_m\) on each \(J_m\). Indeed, this function is concave because each linear segment's slope lies between the corresponding endpoint derivatives of \(R_0\). It majorizes \(R\), and any concave majorant must lie above the line joining the unchanged endpoints of each \(J_m\). Since \(R<L_m\) in the interior of \(J_m\), ironing is strict there. The corresponding intervals in value space are
\[
I_m=[\ell_m,h_m]
=[v(2u_m),v(u_m)]
=[1-\sqrt{2u_m},\,1-\sqrt{u_m}].
\]
They satisfy \(0<\ell_m<h_m<\ell_{m+1}<1\) and \(\ell_m\to1\).

\emph{Property $(iii)$.} The connection between the revenue curve and virtual values is given, wherever the relevant derivatives exist, by
\[
R'(u)=v(u)+uv'(u)
=v(u)-\frac{u}{f(v(u))}
=\varphi_F(v(u)).
\]
Similarly, the ironed virtual value at \(v(u)\) is the slope of the concave envelope of \(R\). On each \(J_m\), this slope is \(A_m>0\), so the ironed virtual value is constant on the interior of \(I_m\), with \(c_m=A_m<c_{m+1}\). Outside these intervals, the envelope has derivative \(1-\tfrac32\sqrt u\), which is strictly decreasing in \(u\). Since \(v(u)\) is strictly decreasing, the ironed virtual value is strictly increasing in value outside the \(I_m\). Finally, the derivative is negative for \(u>4/9\), corresponding to values in \((0,1/3)\).

\emph{Property $(iv)$.} Since \(F(v(u))=1-u\), we have \(\rho_m=u_m\) and \(F(\ell_m)=1-2u_m\to1\). Moreover,
\[
\frac{h_m-\ell_m}{\rho_m}
=\frac{\sqrt2-1}{\sqrt{u_m}}\to\infty.
\]
\end{proof}

\emph{Step 1: fixing a prior.}
Let \(X^0=(\mathcal A^0,\Phi^0)\in\mathcal C\) be a standard,
payment-regular auction. Suppose toward a contradiction that, for every symmetric IPV prior, some mechanism in \(\mathcal C\) admits a revenue-maximizing symmetric equilibrium. Fix a distribution \(F\) satisfying Lemma~\ref{lemma:distconstr}.\footnote{It is worth noting that alternative proofs using distributions with a single ironing interval are possible, but the arguments we have developed require stronger technical assumptions.} Then some symmetric mechanism \(X^1=(\mathcal A^1,\Phi^1)\in\mathcal C\) has a symmetric equilibrium \(\sigma\) that attains optimal revenue under \(F\). For \(r\in\{0,1\}\), write \(q_i^r(a)\) and \(p_i^r(a)\) for bidder \(i\)'s allocation probability and expected payment in \(X^r\).

\emph{Step 2: equilibrium allocation under \(X^1\).}
Let \(Q_i^{\mathrm{eq}}(\theta_i)\) denote bidder \(i\)'s interim allocation probability when her type is \(\theta_i\) and bidders play \(\sigma\) in \(X^1\). Since this equilibrium attains optimal revenue and bidders can choose not to participate, the optimality conditions of \citet{myerson1981optimal} imply:
\begin{enumerate}
\item For almost every type profile at which the highest ironed virtual value is strictly positive, the object is allocated with probability one among bidders attaining that value.
\item For each bidder \(i\), \(Q_i^{\mathrm{eq}}\) is constant for almost every type in each \(I_m\), because ironing is strict throughout its interior.
\end{enumerate}

\emph{Step 3: nonparticipation eliminates the additive payoff term.}
Strategic analogy of \(X^0\) and \(X^1\) gives bijections
\(\alpha_i:\mathcal A_i^1\to\mathcal A_i^0\) and
\(\tau_i:\mathbb R\to\mathbb R\), and functions \(\kappa_i>0\) and
\(\lambda_i\), such that, for every \(i\), \(\theta\), and \(a\),
\begin{equation}\label{eq:ss-payoff-myerson}
\tau_i(\theta)q_i^0(\alpha(a))-p_i^0(\alpha(a))
=\kappa_i(\theta)\bigl(\theta q_i^1(a)-p_i^1(a)\bigr)
+\lambda_i(\theta).
\end{equation}
Let \(\bot_i\) be a nonparticipation action in \(X^1\). If
\(\alpha_i(\bot_i)\) were an active bid \(b\), varying the other
bidders' relabeled actions between unanimous nonparticipation and
unanimous bids of \(b\) would change bidder \(i\)'s allocation from
\(1\) to \(1/N\). At both profiles, however, the right-hand side of
\eqref{eq:ss-payoff-myerson} would equal \(\lambda_i(\theta)\).
Subtracting would force \((1-1/N)\tau_i(\theta)\) to be constant,
contradicting surjectivity of \(\tau_i\). Therefore
\(\alpha_i(\bot_i)=\varnothing\), and substitution into
\eqref{eq:ss-payoff-myerson} gives \(\lambda_i\equiv0\).

Fix bidder \(i\). Evaluating \eqref{eq:ss-payoff-myerson} at
\(\theta=0\) and \(\theta=1\) gives, for every action profile \(a\),
\begin{equation}\label{eq:myerson-payment-zero}
p_i^0(\alpha(a))
=\tau_i(0)q_i^0(\alpha(a))+\kappa_i(0)p_i^1(a),
\end{equation}
\begin{equation}\label{eq:myerson-allocation-one}
q_i^1(a)
=\frac{\tau_i(1)-\tau_i(0)}{\kappa_i(1)}q_i^0(\alpha(a))
+\frac{\kappa_i(1)-\kappa_i(0)}{\kappa_i(1)}p_i^1(a).
\end{equation}
Both \(a\mapsto q_i^0(\alpha(a))\) and \(a\mapsto q_i^1(a)\)
attain the value of \(1\). For the former, this follows from \(X^0\) being a standard auction and from \(\alpha\) being surjective. For the latter, it follows
because, under \(\sigma\), bidder \(i\) receives the object with
probability one for almost every type profile at which she alone
has a strictly positive ironed virtual value, an event of positive
probability.

\emph{Step 4: allocations and payments are linearly independent.}
We now show that \(a\mapsto q_i^0(\alpha(a))\) and
\(a\mapsto p_i^1(a)\) are linearly independent. Suppose otherwise.
Since the former is not identically zero, there is a constant \(e\)
such that \(p_i^1(a)=e q_i^0(\alpha(a))\) for every \(a\).
Equations~\eqref{eq:myerson-payment-zero} and
\eqref{eq:myerson-allocation-one} then imply
\[
p_i^0(\alpha(a))=d q_i^0(\alpha(a)),
\qquad
q_i^1(a)=D q_i^0(\alpha(a))
\]
for constants \(d,D\). Since both allocation probabilities lie in
\([0,1]\) and attain \(1\), we must have \(D=1\). Surjectivity of
\(\alpha\) therefore implies that bidder \(i\)'s payoff in \(X^0\)
at any action profile \(b\) when her type is $t$ is
\[
tq_i^0(b)-p_i^0(b)=(t-d)q_i^0(b).
\]
Against any distribution of opposing bids, bidding arbitrarily high
makes the winning probability approach one, while nonparticipation
gives allocation zero. Thus, a best response with interim allocation
strictly between zero and one is possible only for \(t=d\).

Under \(\sigma\), however, almost every type in each \(I_m\) has
interim allocation strictly between zero and one: there is positive
probability that all opponents have lower ironed virtual values,
and positive probability that an opponent has a higher ironed
virtual value. By \eqref{eq:ss-payoff-myerson}, the relabeled
equilibrium strategy of each such type \(\theta\) is a best response
in \(X^0\) for type \(\tau_i(\theta)\), against the relabeled
opponents' play. Since allocations coincide, it has the same
interim allocation. Hence \(\tau_i(\theta)=d\) for almost every
\(\theta\in I_m\), contradicting injectivity of \(\tau_i\).

\emph{Step 5: allocations coincide, and payments and types transform affinely.}
Substituting \eqref{eq:myerson-payment-zero} and
\eqref{eq:myerson-allocation-one} into
\eqref{eq:ss-payoff-myerson} and using linear independence to
compare coefficients on \(p_i^1(a)\) yields
\[
\kappa_i(\theta)
\left(1-\frac{\kappa_i(1)-\kappa_i(0)}{\kappa_i(1)}\theta\right)
=\kappa_i(0)
\qquad\text{for every }\theta\in\mathbb R.
\]
Since \(\kappa_i(\theta)>0\) for every \(\theta\in\mathbb R\),
the coefficient of \(\theta\) in parentheses must vanish.
Thus \(\kappa_i\equiv\kappa_i(0)\), and
\eqref{eq:myerson-allocation-one} becomes
\[
q_i^1(a)
=\frac{\tau_i(1)-\tau_i(0)}{\kappa_i(0)}q_i^0(\alpha(a)).
\]
Since both allocation probabilities lie in \([0,1]\) and attain
\(1\), the proportionality factor equals \(1\). Hence
\(q_i^1(a)=q_i^0(\alpha(a))\) for every \(a\).

Writing \(c_i:=\kappa_i(0)>0\) and \(d_i:=\tau_i(0)\), equations~\eqref{eq:ss-payoff-myerson} and \eqref{eq:myerson-payment-zero} then give
\begin{equation}\label{eq:anchored-affine-myerson}
q_i^0(\alpha(a))=q_i^1(a),
\quad \tau_i(\theta)=c_i\theta+d_i, \quad
p_i^0(\alpha(a))=c_i p_i^1(a)+d_iq_i^1(a).
\end{equation}

\emph{Step 6: each ironing interval produces one bid atom.}
The action maps are common across bidders. Indeed, symmetry of \(X^1\)
implies equal allocation probabilities at any unanimous action profile.
By \eqref{eq:anchored-affine-myerson}, the relabeled profile in
\(X^0\) also gives every bidder the same allocation probability.
Under the standard allocation rule this is possible only if every
relabeled action is nonparticipation or all are the same active bid.
Thus \(\alpha_i=\alpha_j\) for all \(i,j\); write \(\alpha\) for
their common map. Opponents' relabeled bids are independent draws from
\(\mu:=\int\alpha_\#\sigma(\theta)\,dF(\theta)\). Let \(Q^0(b)\)
denote the interim allocation from relabeled action \(b\) against
\(\mu^{N-1}\).

By the equality condition in the ironing revenue bound, each bidder's
equilibrium interim allocation is constant, say
\(\bar Q_m\in(0,1)\), for \(F\)-almost every \(\theta\in I_m\).
Positivity follows from the positive probability that every opponent
has smaller ironed virtual value; the upper bound follows from the
positive probability of an opponent having larger ironed virtual
value. Moreover, almost every pure action used by these types has that
same interim allocation. To see this, equilibrium utility is finite
because nonparticipation yields zero and payments are nonnegative.
A best-response mixture therefore uses optimal pure actions almost
surely. For an interior type \(\theta\) and such an action with
interim allocation \(Q\), comparison with equilibrium mixtures of
types \(\theta^-<\theta<\theta^+\) in \(I_m\), both having
allocation \(\bar Q_m\), gives
\[
(\theta-\theta^-)(Q-\bar Q_m)\ge0,
\qquad
(\theta^+-\theta)(\bar Q_m-Q)\ge0.
\]
Hence \(Q=\bar Q_m\), as claimed.

Consider the conditional bid distribution
\(\mu^m:=\int_{I_m}\alpha_\#\sigma(\theta)\,dF(\theta\mid I_m)\).
It assigns no mass to nonparticipation, and
\(Q^0(b)=\bar Q_m\) for \(\mu^m\)-almost every \(b\). If
\(\mu^m\) were not a point mass, we could choose \(b^-<b^+\) with
these equal interim allocations and \(\mu^m([b^-,b^+])>0\). Then
\(\mu([b^-,b^+])>0\), so independence gives positive probability
that the highest opposing bid lies in \([b^-,b^+]\). Raising the bid
from \(b^-\) to \(b^+\) weakly increases allocation at every profile
and strictly increases it on this event, including its endpoints under
uniform tie-breaking. This contradicts equality of the interim
allocations. Therefore \(\mu^m=\delta_{b_m}\) for some active bid
\(b_m\).

Conditional on all opponents having types in \(I_m\), almost every
type strictly below \(\ell_m\) must lose, whereas almost every type
strictly above \(h_m\) must win: their ironed virtual values are
strictly below and above that on \(I_m\), respectively. It follows
that almost all lower types use relabeled actions below \(b_m\),
including nonparticipation, and almost all higher types use bids
strictly above \(b_m\). In particular,
\(\mu(\{b_m\})=\rho_m\) and
\(s_m:=\mu(\{\varnothing\}\cup[\underline a,b_m))=F(\ell_m)\).
The bids \(b_m\) are strictly increasing. For all sufficiently large
\(m\), they are bounded away from \(\underline a\): fix any active
pooling bid from an earlier interval and use the next, strictly higher
pooling bid as a common lower bound.

\emph{Step 7: payments remain uniformly bounded.}
Let \(Q^{\mathrm{eq}}(\theta)\) denote a bidder's equilibrium interim
allocation in \(X^1\). Incentive compatibility implies that it is
weakly increasing. Choose separating types \(\theta_n\uparrow1\)
at which the revenue-optimal allocation characterization holds.
Then \(Q^{\mathrm{eq}}(\theta_n)=F(\theta_n)^{N-1}\to1\).
Since equilibrium optimality is required also at \(\theta=1\),
monotonicity gives \(Q^{\mathrm{eq}}(1)=1\). By
\eqref{eq:anchored-affine-myerson}, almost every relabeled bid used
by that type also has interim allocation one. Choose one such finite
active bid \(B\). Uniform tie-breaking implies
\(\mu([B,\infty))=0\); otherwise there would be a positive
probability of losing or tying when bidding \(B\).

All opponents' relabeled actions thus lie almost surely in \(\{\varnothing\}\cup[\underline a,B]\). Fix bidder \(i\). We now show that her payments remain bounded as we vary her bid and her opponents' actions within these bounds. To that end, choose a compact interval strictly above \(\underline a\) containing \(b_m\) and \(b_m\pm\epsilon\) for all sufficiently large \(m\) and sufficiently small \(\epsilon>0\). Since \(p_i^0\) is jointly locally bounded and there are finitely many possible participant sets, compactness gives a uniform bound \(0\le p_i^0\le M_i<\infty\) whenever bidder \(i\)'s bid lies in this interval and all opponents' actions lie in \(\{\varnothing\}\cup[\underline a,B]\).

\emph{Step 8: a small bid atom can support only a short interval of types.}
Fix a sufficiently large \(m\). Let \(T_m\) be the event that the
highest active opposing bid equals \(b_m\), and write
\(q_m(z):=q_i^0(b_m,z)\). On \(T_m\), define the payment difference
\(\Delta_m(z):=\pi_{i,w}^0(b_m,z)-\pi_{i,\ell}^0(b_m,z)\), which
may depend on the entire opposing profile \(z\). Own-bid continuity and the uniform bound from Step~7 imply
\(0\le\pi_{i,w}^0(b_m,z),\pi_{i,\ell}^0(b_m,z)\le M_i\),
by approaching the tied bid from above and below, respectively.
Hence \(\lvert\Delta_m(z)\rvert\le M_i\).

For almost every \(\theta\in I_m\), the bid \(b_m\) is a best
response in \(X^0\) for type \(t=c_i\theta+d_i\). Compare it with
\(b_m+\epsilon\) and \(b_m-\epsilon\). Off \(T_m\), sufficiently
small deviations preserve winning or losing status, so the payoff
differences converge to zero. On \(T_m\), the limits of the payoff
from \(b_m\) minus the payoff from the upward and downward deviations
are, respectively,
\[
(1-q_m(z))(\Delta_m(z)-t),
\qquad
q_m(z)(t-\Delta_m(z)).
\]
Dominated convergence applies because the absolute payoff differences
are bounded by \(\lvert t\rvert+2M_i\). Writing expectations under
\(\mu^{N-1}\), the two incentive inequalities therefore give
\(D_m^-\le t\le D_m^+\), where
\[
D_m^-:=
\frac{\mathbb E[\mathbf1_{T_m}q_m\Delta_m]}
     {\mathbb E[\mathbf1_{T_m}q_m]},
\qquad
D_m^+:=
\frac{\mathbb E[\mathbf1_{T_m}(1-q_m)\Delta_m]}
     {\mathbb E[\mathbf1_{T_m}(1-q_m)]}.
\]
Both denominators are positive, and \(\lvert D_m^-\rvert\le M_i\).

On the event that exactly one opponent bids \(b_m\) and all others
bid below it, including nonparticipation, we have \(q_m=1/2\).
This event has probability \(e_m=(N-1)\rho_m s_m^{N-2}\). Let
\(z_m\) denote the probability of the remaining part of \(T_m\),
on which at least two opponents bid \(b_m\). A union bound gives
\(z_m\le\binom{N-1}{2}\rho_m^2\). Since \(1-2q_m=0\) on the
first event and \(0\le1-2q_m\le1\) on the second,
\[
D_m^+-D_m^-
=
\frac{\mathbb E[\mathbf1_{T_m}(1-2q_m)(\Delta_m-D_m^-)]}
     {\mathbb E[\mathbf1_{T_m}(1-q_m)]}\le \frac{4M_i z_m}{e_m}
\le \frac{2M_i(N-2)\rho_m}{s_m^{N-2}}.
\]
For \(N=2\), the bound is zero. For every \(N\ge2\), we have
\(s_m\to1\), so there is a finite constant \(C_i\), independent
of \(m\), such that \(D_m^+-D_m^-\le C_i\rho_m\) for all
sufficiently large \(m\).

The incentive bounds \(D_m^-\le c_i\theta+d_i\le D_m^+\) hold for
almost every type in \(I_m\). Since \(F\) has a positive density
almost everywhere there, these types have images under \(\tau_i\)
arbitrarily close to both endpoints of
\([c_i\ell_m+d_i,c_i h_m+d_i]\). It follows that
\(c_i(h_m-\ell_m)\le D_m^+-D_m^-\le C_i\rho_m\), contradicting
property $(iv)$ because \(c_i>0\). This
completes the proof.

\subsection{Proof of Proposition~\ref{prop:procurement-scores}}

We first consider the linear-score family. Fix
\(\lambda_1,\lambda_2\in(0,1)\). For each seller \(i\), define
\[
\alpha_i(b_i):=b_i+\left(\frac{1}{\lambda_2}-\frac{1}{\lambda_1}\right)q_i,
\qquad
\tau_i(c_i):=c_i-\left(\frac{1}{\lambda_2}-\frac{1}{\lambda_1}\right)q_i,
\qquad
\kappa_i(c_i):=1,
\qquad
\lambda_i(c_i):=0.
\]
Each \(\alpha_i:\mathbb R\to\mathbb R\) and each
\(\tau_i:\mathbb R\to\mathbb R\) is a bijection.

For every bid profile \(b\),
\[
S^{L}_{\lambda_2}(\alpha(b))_i=
-\lambda_2\left(b_i+\left(\frac{1}{\lambda_2}-\frac{1}{\lambda_1}\right)q_i\right)+(1-\lambda_2)q_i  =
\frac{\lambda_2}{\lambda_1}
\left[-\lambda_1 b_i+(1-\lambda_1)q_i\right]  =
\frac{\lambda_2}{\lambda_1}S^L_{\lambda_1}(b)_i .
\]
Thus all sellers' scores are multiplied by the same positive constant. Hence the set of highest-scoring sellers, and therefore the winner selected by the fixed tie-breaking rule, is unchanged:
\[
w_{L(\lambda_2)}(\alpha(b))=w_{L(\lambda_1)}(b).
\]
Therefore, for every seller \(i\), every cost \(c_i\), and every bid profile
\(b\),
\[
\begin{aligned}
U_i^{L(\lambda_2)}[c_i,\alpha(b)]
&=
\mathbf{1}\{w_{L(\lambda_1)}(b)=i\}
\left(b_i+\left(\frac{1}{\lambda_2}-\frac{1}{\lambda_1}\right)q_i-c_i\right)  \\
&=
\mathbf{1}\{w_{L(\lambda_1)}(b)=i\}
\Big(b_i-\Big(c_i-\left(\frac{1}{\lambda_2}-\frac{1}{\lambda_1}\right)q_i\Big)\Big) =
U_i^{L(\lambda_1)}[\tau_i(c_i),b].
\end{aligned}
\]
This is the affine strategic analogy condition with
\(\kappa_i(c_i)=1\) and \(\lambda_i(c_i)=0\). Hence
\(\{L(\lambda):\lambda\in(0,1)\}\) is a class of strategically analogous mechanisms.

Now consider the ratio-score family. Suppose the fixed quality scores are not all equal. Choose sellers \(1\) and \(2\) such that $q_1=\max_j q_j$ and $q_1>q_2.$ Choose \(\lambda_1,\lambda_2\in(0,1)\) such that
\begin{equation}\label{eq:lambdasdefine}   
\lambda_1<\frac{(q_1-q_2)}{1+(q_1-q_2)} <\lambda_2.
\end{equation}

We first show that seller \(2\) can never win in \(R(\lambda_1)\). For any bid profile \(b\in\mathbb{R}_{++}^N\), we have ${\min_j b_j}/{b_2}\le 1$, and so seller \(2\)'s score is at most $\lambda_1+(1-\lambda_1)q_2.$ Because \(\lambda_1<\frac{(q_1-q_2)}{1+(q_1-q_2)}\), we have $\lambda_1+(1-\lambda_1)q_2
<
(1-\lambda_1)q_1.$ Moreover, note that seller \(1\)'s score is strictly larger than \((1-\lambda_1)q_1\). Thus seller \(2\)'s score is always strictly below seller \(1\)'s score, and so seller \(2\) never wins in \(R(\lambda_1)\).

Next we show that seller \(2\) can win in \(R(\lambda_2)\). By \eqref{eq:lambdasdefine}, we then have $\lambda_2+(1-\lambda_2)q_2
>
(1-\lambda_2)q_1.$ Choose \(b_2>0\); we can then make all \(b_j\) for $j\neq 2$ sufficiently large that ${\min_k b_k}/{b_2}=1$ and
\[
\frac{\min_k b_k}{b_j}
=
\frac{b_2}{b_j}
\]
is arbitrarily close to \(0\) for every \(j\neq 2\). Hence, seller \(2\)'s score is $\lambda_2+(1-\lambda_2)q_2,$ while every other seller's score is arbitrarily close to $(1-\lambda_2)q_j
\le
(1-\lambda_2)q_1.$ Thus, for \(b_j\) large enough for all \(j\neq 2\), seller \(2\) is the unique winner in \(R(\lambda_2)\).

We now prove that \(R(\lambda_1)\) and \(R(\lambda_2)\) cannot be strategically analogous. Suppose, toward a contradiction, that they are strategically analogous. Then there exist bijections \(\alpha_i:\mathbb{R}_{++}\to\mathbb{R}_{++}\), bijections \(\tau_i:\mathbb R\to\mathbb R\), and functions
\(\kappa_i:\mathbb R\to\mathbb{R}_{++}\) and
\(\lambda_i:\mathbb R\to\mathbb R\) such that, for every seller \(i\), every cost
\(c_i\), and every bid profile \(b\),
\[
U_i^{R(\lambda_2)}[c_i,\alpha(b)]
=
\kappa_i(c_i)U_i^{R(\lambda_1)}[\tau_i(c_i),b]+\lambda_i(c_i).
\]

Apply this to seller \(2\). Since seller \(2\) never wins in \(R(\lambda_1)\), we have $U_2^{R(\lambda_1)}[\tau_2(c_2),b]=0$ for every $c_2$ and every $b.$ Therefore $U_2^{R(\lambda_2)}[c_2,\alpha(b)]
=
\lambda_2(c_2)$ for every $c_2$ and every $b$. Since \(\alpha\) is onto, this implies that, for every fixed \(c_2\),
seller \(2\)'s payoff in \(R(\lambda_2)\) is constant across all bid profiles, which is not the case. In \(R(\lambda_2)\), seller \(2\) loses at some bid profiles, so her payoff is \(0\). She also wins at some bid profiles. Taking \(c_2=0\) and a winning bid profile with \(b_2>0\), her payoff is \(b_2>0\). Thus her payoff is not constant across bid profiles, a contradiction.

\subsection{Proof of Proposition \ref{prop:pricing-analogy}}

\emph{Input-based pricing.} Fix two efficacy levels \(e,e'>0\). Write $\eta:=\frac{e'}{e}.$ For each buyer \(i\), choose the action and type bijections $\alpha_i^{e\to e'}(a_i):=a_i$ and $\tau_i^{e'\to e}(v_i)(x):=v_i(\eta x).$ The former is a bijection because the action space is \([0,1]\) in every efficacy
state. The latter is well-defined because if $v_i\in \mathcal V$ then $\tau_i^{e'\to e}(v_i) \in \mathcal V$, and is a bijection because it is invertible via $(\tau_i^{e'\to e})^{-1}(w_i)(x)=w_i(x/\eta).$ Take any action profile \(a\in [0,1]^{\mathcal I}\). Since the action relabeling is the
identity, the requested capacities are unchanged. Therefore the rationing
factor $\min \{1,{1}/{\sum_{j\in\mathcal I} a_j}\}$ is unchanged, and so the allocated capacity
\[
z_i(a)= \begin{cases} a_i\,\min\left\{1,\,1\big/\textstyle\sum_{j\in\mathcal I}a_j\right\}, & \sum_{j\in\mathcal I}a_j>0,\\ 0,& \text{otherwise}, \end{cases}
\]
is the same under efficacy \(e\) and \(e'\). Hence buyer \(i\)'s output
scales with efficacy:
\[
x_i(\alpha^{e\to e'}(a),e')
=
e'\,z_i(a)
=
\eta\, e\,z_i(a)
=
\eta\, x_i(a,e).
\]
Payments are unchanged because the tariff depends only on the allocated
capacity:
\[
p_i(\alpha^{e\to e'}(a))
=
P(z_i(a))
=
p_i(a).
\]

Now fix buyer \(i\), type \(v_i\in\mathcal V\), and action profile
\(a\in [0,1]^{\mathcal I}\). Using the identities above,
\[
U_i^{X^{\mathrm{in}}_{e'}}[v_i,\alpha^{e\to e'}(a)]
=
v_i\left(x_i(\alpha^{e\to e'}(a),e')\right)
-
p_i(\alpha^{e\to e'}(a)) =
v_i(\eta\, x_i(a,e))-P(z_i(a)),
\]
\[
U_i^{X^{\mathrm{in}}_e}[\tau_i^{e'\to e}(v_i),a]
=
\tau_i^{e'\to e}(v_i)\left(x_i(a,e)\right)
-
p_i(a) =
v_i(\eta\, x_i(a,e))-P(z_i(a)).
\]
Therefore, $U_i^{X^{\mathrm{in}}_{e'}}[v_i,\alpha^{e\to e'}(a)]
=
U_i^{X^{\mathrm{in}}_e}[\tau_i^{e'\to e}(v_i),a].$ Thus \(X^{\mathrm{in}}_e\) and \(X^{\mathrm{in}}_{e'}\) are strategically analogous with $\kappa_i(v_i)=1$ and $\lambda_i(v_i)=0.$

\emph{Output-based pricing.} We show that
\(\{X^{\mathrm{out}}_e:e>0\}\) is a class of strategically analogous
mechanisms if and only if \(P(x)=Ax^\rho\) for some \(A>0\) and
\(\rho>1\). We first prove necessity. Suppose it is a family of strategically analogous mechanisms. Fix two efficacy levels \(e,e'>0\), and let \(X^{\mathrm{out}}_e\) and \(X^{\mathrm{out}}_{e'}\) denote their corresponding mechanisms. Then there exist bijections
\(\alpha_i:[0,e]\to[0,e']\), \(\tau_i:\mathcal V\to\mathcal V\) and functions
\(\kappa_i:\mathcal V\to\mathbb R_{++}\),
\(\lambda_i:\mathcal V\to\mathbb R\) such that, for every buyer \(i\), every
\(v_i\in\mathcal V\), and every \(a\in[0,e]^N\),
\begin{equation}\label{eq:raw-affine-sim}
U_i^{X^{\mathrm{out}}_{e'}}[v_i,\alpha(a)]
=
\kappa_i(v_i)\,U_i^{X^{\mathrm{out}}_e}[\tau_i(v_i),a]
+
\lambda_i(v_i).
\end{equation}

For an action profile \(a\in[0,e]^N\), write \(x_i(a,e)\) for buyer \(i\)'s
allocation in efficacy state \(e\).

We first show that \(x_i(\alpha(a),e')\) depends on \(a\) only through \(x_i(a,e)\), i.e.\ that there exists a function \(h_i:[0,e]\to[0,e']\) such that
\begin{equation}\label{eq:h-def}
x_i(\alpha(a),e') = h_i\bigl(x_i(a,e)\bigr)
\qquad\text{for every } a\in[0,e]^N.
\end{equation}

Let \(0\in\mathcal V\) denote a constant zero valuation. Applying
\eqref{eq:raw-affine-sim} to \(v_i\equiv0\) gives
\[
-P(x_i(\alpha(a),e'))
=
\kappa_i(0)\left(\tau_i(0)(x_i(a,e))-P(x_i(a,e))\right)
+
\lambda_i(0).
\]
The right-hand side depends on \(a\) only through \(x_i(a,e)\). Since \(P\) is
strictly increasing, it follows that \(x_i(\alpha(a),e')\) also depends on
\(a\) only through \(x_i(a,e)\), establishing \eqref{eq:h-def}.

\begin{fact}\label{fact:alpha-zero}
\(\alpha_i(0)=0\) for every \(i\).
\end{fact}

\begin{proof}
Fix \(i\). If \(a_i=0\), then \(x_i(a,e)=0\) for every \(a_{-i}\). Therefore,
by \eqref{eq:h-def}, $x_i((\alpha_i(0),\alpha_{-i}(a_{-i})),e')$ is independent of \(a_{-i}\). Suppose \(\alpha_i(0)>0\). Choose some
\(j\neq i\). For each \(\ell\neq i,j\), choose \(a_\ell\in[0,e]\) such that
\(\alpha_\ell(a_\ell)=0\), which is possible because \(\alpha_\ell\) is onto.
Now vary \(a_j\) so that \(\alpha_j(a_j)\) ranges over all of \([0,e']\). Then
buyer \(i\)'s allocation under efficacy \(e'\) is
\[
\alpha_i(0)
\min\left\{
1,\frac{e'}{\alpha_i(0)+\alpha_j(a_j)}
\right\}.
\]
This expression is not constant as \(\alpha_j(a_j)\) varies from \(0\) to
\(e'\), because \(\alpha_i(0)>0\). This contradicts the independence from
\(a_{-i}\).
\end{proof}

Now take a profile in which buyer \(i\) requests \(x\in[0,e]\) and every other
buyer requests \(0\). Since \(\alpha_j(0)=0\) for every \(j\neq i\), buyer
\(i\)'s allocation under efficacy \(e'\) is \(\alpha_i(x)\). Therefore $h_i(x)=\alpha_i(x).$ Thus \eqref{eq:h-def} becomes
\begin{equation}\label{eq:allocation-relabeling}
x_i(\alpha(a),e')
=
\alpha_i(x_i(a,e))
\qquad
\text{for every }a\in[0,e]^N.
\end{equation}

Normalize the action relabelings by defining $f_i(z):=\frac{\alpha_i(ez)}{e'}$ for $z\in[0,1].$ Each \(f_i\) is a bijection from \([0,1]\) to \([0,1]\), and \(f_i(0)=0\).

\begin{fact}\label{fact:f-increasing}
Each \(f_i\) is strictly increasing.
\end{fact}

\begin{proof}
Fix \(0\leq a<b\leq 1\).
If \(a=0\), then \(f_i(a)=0<f_i(b)\), since \(f_i\) is injective and
\(f_i(0)=0\). So suppose \(a>0\). Choose \(j\neq i\) and choose $c\in(1-b,1-a).$ Consider profiles in which only buyers \(i\) and \(j\) make positive requests.

At normalized requests \((a,c)\), total demand is below capacity, so buyer
\(i\)'s output under efficacy \(e\) is \(ea\). By
\eqref{eq:allocation-relabeling}, the relabeled allocation under efficacy
\(e'\) must be \(\alpha_i(ea)\). Since \(a>0\), we have
\(\alpha_i(ea)>0\). If the relabeled profile were rationed, buyer \(i\)'s
allocation would be strictly below her relabeled request \(\alpha_i(ea)\), a
contradiction. Hence the relabeled profile is not rationed, and therefore $f_i(a)+f_j(c)\leq 1.$

At normalized requests \((b,c)\), total demand exceeds capacity, so buyer
\(i\)'s output under \(e\) is $e\frac{b}{b+c}<eb.$ By \eqref{eq:allocation-relabeling}, the relabeled allocation under efficacy
\(e'\) must be $\alpha_i\left(e\frac{b}{b+c}\right).$ If the relabeled profile were not rationed, buyer \(i\)'s allocation would be
her relabeled request \(\alpha_i(eb)\). But $e\frac{b}{b+c}<eb$ and \(\alpha_i\) is injective, so $\alpha_i\left(e\frac{b}{b+c}\right)\neq \alpha_i(eb).$ Hence the relabeled profile must be rationed, and therefore $f_i(b)+f_j(c)>1.$ Combining the two inequalities gives $f_i(a)<f_i(b)$.
\end{proof}

Since \(f_i\) is a strictly increasing
bijection from \([0,1]\) onto \([0,1]\), it is continuous and satisfies
\(f_i(1)=1\). Now, fix distinct buyers \(i\neq j\) and note two implications of
\eqref{eq:allocation-relabeling}. First, if \(a,b\in[0,1]\) and \(a+b=1\),
then
\begin{equation}\label{eq:boundary-f}
f_i(a)+f_j(b)=1.
\end{equation}
Indeed, if \(a,b\in(0,1)\), the source profile is unrationed, so the relabeled
profile is unrationed and $f_i(a)+f_j(b)\leq 1.$
For every small \(\varepsilon>0\), the source profile \((a+\varepsilon,b)\) is
rationed, so the relabeled profile is rationed and $f_i(a+\varepsilon)+f_j(b)>1.$ Letting \(\varepsilon\downarrow 0\) gives $f_i(a)+f_j(b)\geq 1.$ Thus equality holds. The endpoint cases \(a=0\) or \(b=0\) follow from
\(f_i(0)=0\) and \(f_i(1)=1\).

Second, if \(a,b\in(0,1]\) and \(a+b>1\), then proportional rationing and
\eqref{eq:allocation-relabeling} imply
\begin{equation}\label{eq:ration-f}
f_i\left(\frac{a}{a+b}\right)
=
\frac{f_i(a)}{f_i(a)+f_j(b)}.
\end{equation}

\begin{fact}\label{fact:f-common}
There is a common increasing bijection \(f:[0,1]\to[0,1]\) such that \(f_i=f\) for every \(i\), with \(f(z)+f(1-z)=1\) for every \(z\in[0,1]\).
\end{fact}

\begin{proof}
By \eqref{eq:boundary-f}, we have $f_i(z)+f_j(1-z)=1$ for every $z\in[0,1].$ Taking \(a=b=s/2\) in \eqref{eq:ration-f}, where \(s\in(1,2]\), gives
\[
f_i(1/2)
=
\frac{f_i(s/2)}{f_i(s/2)+f_j(s/2)}
\quad \implies
\quad
\frac{f_i(s/2)}{f_j(s/2)}
=
\frac{f_i(1/2)}{1-f_i(1/2)}.
\]
At \(s=2\), the left-hand side equals \(1\), because \(f_i(1)=f_j(1)=1\).
Therefore $f_i(z)=f_j(z)$ for every $z\in[1/2,1].$ Using $f_i(z)+f_j(1-z)=1$ and $f_j(z)+f_i(1-z)=1,$ this equality extends to all \(z\in[0,1]\).
\end{proof}

\begin{fact}\label{fact:f-identity}
\(f(z)=z\) for every \(z\in[0,1]\).
\end{fact}

\begin{proof}
With a common \(f\), equation \eqref{eq:ration-f} implies that, for all
\(a,b\in(0,1]\) with \(a+b>1\),
\[
\frac{
f\left(\frac{a}{a+b}\right)
}{
1-f\left(\frac{a}{a+b}\right)
}
=
\frac{f(a)}{f(b)}.
\]
Now take \(a=zy\) and \(b=y\), where \(z,y\in(0,1]\) and \(y(1+z)>1\). Then
\[
\frac{
f\left(\frac{z}{1+z}\right)
}{
1-f\left(\frac{z}{1+z}\right)
}
=
\frac{f(zy)}{f(y)}
\quad \implies \quad 
\frac{
f\left(\frac{z}{1+z}\right)
}{
1-f\left(\frac{z}{1+z}\right)
}
=
f(z).
\]
Therefore
\begin{equation}\label{eq:local-multiplicative-f}
f(zy)=f(z)f(y)
\end{equation}
whenever \(z,y\in(0,1]\) and \(y(1+z)>1\). This extends to all \(z,y\in(0,1]\). Fix
\(z,y\in(0,1]\). Choose \(m\) large enough that $y^{1/m}(1+zy)>1.$ Then, for each \(\ell=0,\dots,m-1\), we have $y^{1/m}\left(1+z y^{\ell/m}\right)>1,$ so \eqref{eq:local-multiplicative-f} gives $f\left(z y^{(\ell+1)/m}\right)
=
f\left(z y^{\ell/m}\right)f\left(y^{1/m}\right).$ Iterating gives $f(zy)=f(z)f(y^{1/m})^m.$ Applying the same argument with \(z=1\) gives $f(y)=f(y^{1/m})^m.$ Hence $f(zy)=f(z)f(y)$ for all $z,y\in(0,1].$ Thus, the function $g(t):=\log f(e^t)$ for $t\leq 0$ is continuous and additive: $g(s+t)=g(s)+g(t)$ for all $s,t\leq 0.$ Therefore there exists \(\rho>0\) such that $f(z)=z^\rho$ for every $z\in(0,1]$. Since \(f(z)+f(1-z)=1\) for every \(z\in[0,1]\), evaluating at \(z=1/2\)
gives $2(1/2)^\rho=1$, and so \(\rho=1\).
\end{proof}
Undoing the normalization then gives $\alpha_i(a_i)=\eta a_i$ for every $i$ and $a_i\in[0,e],$ where $\eta:=\frac{e'}{e},$ so
\begin{equation}\label{eq:allocation-scale}
x_i(\alpha(a),e')=\eta\, x_i(a,e)
\qquad
\text{for every }i\text{ and every }a\in[0,e]^N.
\end{equation}

We now use the convexity of \(P\). Fix buyer \(i\). Applying
\eqref{eq:raw-affine-sim} to $v\equiv 0$ and using
\eqref{eq:allocation-scale} gives
\[
-P(\eta x)
=
\kappa_i(0)\left(\tau_i(0)(x)-P(x)\right)
+
\lambda_i(0)
\qquad
\text{for every }x\in[0,e].
\]
At \(x=0\), we get \(\lambda_i(0)=0\). Therefore, for $c:=\frac{1}{\kappa_i(0)}>0,$ we have
\begin{equation}\label{eq:first-concave}
\tau_i(0)(x)=P(x)-cP(\eta x)
\qquad
\text{for every }x\in[0,e].
\end{equation}
Since \(\tau_i(0)\in\mathcal V\), the function $x\mapsto P(x)-cP(\eta x)$ is weakly increasing and weakly concave on \([0,e]\). Because \(\tau_i\) is onto, there exists \(\bar v_i\in\mathcal V\) such that $\tau_i(\bar v_i)=0.$ Applying \eqref{eq:raw-affine-sim} to \(\bar v_i\), again using
\eqref{eq:allocation-scale}, gives $\bar v_i(\eta x)-P(\eta x)
=
-\kappa_i(\bar v_i)P(x)
+
\lambda_i(\bar v_i)$ for every $x\in[0,e].$ At \(x=0\), we get \(\lambda_i(\bar v_i)=0\). Therefore, for $d:=\kappa_i(\bar v_i)>0,$ we have
\begin{equation}\label{eq:second-concave}
\bar v_i(\eta x)=P(\eta x)-dP(x)
\qquad
\text{for every }x\in[0,e].
\end{equation}
Since \(\bar v_i\in\mathcal V\), the function $x\mapsto P(\eta x)-dP(x)$ is weakly increasing and weakly concave on \([0,e]\).

We use the following fact.

\begin{fact}\label{fact:P-proportional}
Let \(P:\mathbb R_+\to\mathbb R_+\) be a regular price scheme.
Fix \(\eta>0\), \(K>0\). If there exist \(c,d>0\) such that both $x\mapsto P(x)-cP(\eta x)$ and $x\mapsto P(\eta x)-dP(x)$ are weakly increasing and weakly concave on \([0,K]\), then $P(\eta x)=dP(x)$ for every $x\in[0,K].$
\end{fact}

\begin{proof}
Let \(p(x)\) denote the right derivative of \(P\) at \(x>0\). Since \(P\) is
convex, \(p\) exists and is weakly increasing. Since \(P\) is strictly
increasing and convex, \(p(x)>0\) for every \(x>0\). Define
\[
F(x):=P(x)-cP(\eta x),
\qquad
G(x):=P(\eta x)-dP(x).
\]
By assumption, \(F\) and \(G\) are weakly increasing and weakly concave on
\([0,K]\). Hence their right derivatives are nonnegative on \((0,K)\). Thus
\[
F'_+(x)=p(x)-c\eta p(\eta x)\geq 0, \quad \quad
G'_+(x)=\eta p(\eta x)-d p(x)\geq 0.
\]
Combining these inequalities gives $p(x)\geq cd\,p(x).$ Since \(p(x)>0\), we get $cd\leq 1.$

Now use concavity. Since \(F\) and \(G\) are concave, their right derivatives
are weakly decreasing. Hence, for any \(0<a<b<K\), $F'_+(b)\leq F'_+(a)$ implies $p(b)-p(a)\leq c\eta\bigl(p(\eta b)-p(\eta a)\bigr).$ Similarly, $G'_+(b)\leq G'_+(a)$ implies $\eta\bigl(p(\eta b)-p(\eta a)\bigr)\leq d\bigl(p(b)-p(a)\bigr).$ Combining these two inequalities yields $p(b)-p(a)\leq cd\bigl(p(b)-p(a)\bigr).$ Because \(P\) is strictly convex, \(p\) is not constant on \((0,K)\). Hence
there exist \(0<a<b<K\) such that \(p(b)>p(a)\), and so $cd\geq 1.$ Together with \(cd\leq 1\), this gives $cd=1.$

Returning to the derivative inequalities, we have $p(x)\geq c\eta p(\eta x)$ and, since \(d=1/c\), we get $\eta p(\eta x)\geq \frac{1}{c}p(x).$ Therefore $p(x)=c\eta p(\eta x)$ for every $x\in(0,K)$. Thus $F'_+(x)=0$ for every $x\in(0,K).$ Since \(F\) is concave, this implies that \(F\) is constant on \([0,K]\).
Because $F(0)=P(0)-cP(0)=0,$ we get $P(x)=cP(\eta x)$ for every $x\in[0,K].$ Since \(cd=1\), this is equivalent to $P(\eta x)=dP(x)$ for every $x\in[0,K].$
\end{proof}

Applying this fact with \(K=e\), we obtain a constant \(C_{\eta,e}>0\) such that $P(\eta x)=C_{\eta,e}\,P(x)$ for every $x\in[0,e].$ Because \(e>0\) was arbitrary, for every \(\eta>0\) there exists a constant
\(C(\eta)>0\) such that
\begin{equation}\label{eq:homogeneous-P}
P(\eta x)=C(\eta)\,P(x)
\qquad
\text{for every }x\geq 0.
\end{equation}
Indeed, the constants obtained from different values of \(e\) agree on
overlapping intervals because \(P(x)>0\) for every \(x>0\). Setting \(x=1\) in \eqref{eq:homogeneous-P} gives $C(\eta)=\frac{P(\eta)}{P(1)}$, and therefore, for every \(x,\eta>0\),
\[
P(\eta x)=\frac{P(\eta)P(x)}{P(1)}.
\]
Define $Q(x):=\frac{P(x)}{P(1)}.$ Then \(Q:\mathbb R_{++}\to\mathbb R_{++}\) is continuous, strictly increasing,
and satisfies $Q(\eta x)=Q(\eta)Q(x)$ for every $x,\eta>0.$ Thus \(Q\) is a continuous multiplicative function on \(\mathbb R_{++}\). Hence
there exists \(\rho\in\mathbb R\) such that $Q(x)=x^\rho.$ Since \(Q\) is strictly increasing, \(\rho>0\). Therefore $P(x)=P(1)\,x^\rho.$ Writing \(A:=P(1)>0\), we get $P(x)=Ax^\rho$
for some \(A>0\) and \(\rho>0\). That is, \(P\) must be
isoelastic. Strict convexity further requires \(\rho>1\). This establishes necessity.

For sufficiency, suppose \(P(x)=Ax^\rho\) for some \(A>0\) and
\(\rho>1\). Fix \(e,e'>0\) and write \(\eta:=e'/e\). For each buyer
\(i\), define $\alpha_i(a_i):=\eta a_i,$ $\tau_i(v_i)(x):=\eta^{-\rho}v_i(\eta x),$ $\kappa_i(v_i):=\eta^\rho,$ and $\lambda_i(v_i):=0.$ The action map \(\alpha_i\) is a bijection from \([0,e]\) to
\([0,e']\). The type map \(\tau_i\) preserves weak monotonicity,
weak concavity, and the normalization at zero, and is a bijection
of \(\mathcal V\) with inverse $\tau_i^{-1}(w_i)(x)=\eta^\rho w_i(x/\eta).$ For every action profile \(a\in[0,e]^N\), proportional rationing gives $x_i(\alpha(a),e')=\eta x_i(a,e).$
Consequently, writing \(x:=x_i(a,e)\), we have
\[
U_i^{X^{\mathrm{out}}_{e'}}[v_i,\alpha(a)]
=v_i(\eta x)-A(\eta x)^\rho
=\eta^\rho\left[\eta^{-\rho}v_i(\eta x)-Ax^\rho\right]
=\eta^\rho U_i^{X^{\mathrm{out}}_e}[\tau_i(v_i),a].
\]
Thus, the affine characterization of strategic analogy is satisfied.


\end{document}